\documentclass[11pt]{article}
\usepackage[utf8]{inputenc}
\usepackage[numbers]{natbib}
\usepackage{times}
\usepackage[a4paper, margin=1in]{geometry}

\usepackage{pgfplots}
\pgfplotsset{width=10cm,compat=1.9}
\usepgfplotslibrary{external}
\title{Stronger Memory-Query Tradeoffs for Convex Optimization: The Limitations of Subquadratic Memory}
\author{\makebox[1.8in]{Michael Menart  \thanks{Department of Computer Science,
University of Toronto, 
Vector Institute, 
\texttt{michael.menart@utoronto.ca}}}
\makebox[1.8in]{Aleksandar Nikolov \thanks{Department of Computer Science,
University of Toronto, 
\texttt{anikolov@cs.toronto.edu}}
}
\makebox[1.8in]{Ohad Shamir \thanks{Weizmann Institute of Science and University of Toronto, 
\texttt{ohad.shamir@weizmann.ac.il}}
}
}
\date{}

\usepackage{todonotes}

\usepackage{macros}

\begin{document}

\maketitle

\begin{abstract}
    We prove two lower bounds for the first order oracle complexity of minimizing a $d$-dimensional $1$-Lipschitz convex function over the unit ball with $m$ bits of memory. We first show that any such (possibly randomized) algorithm must make $\tilde{\Omega}(\frac{d^2}{\sqrt{m}})$ oracle queries. For deterministic optimization algorithms, we show that $\tilde{\Omega}(\min\{d^{1.6},\frac{d^{8/3}}{m^{2/3}}\})$ queries are required. For all memory regimes of interest, these improves upon the previous best known lower bounds of $\tilde{\Omega}(\max\{\frac{d^{8/3}}{m^{4/3}},\frac{d^{4/3}}{m^{1/6}}\})$ and $\tilde{\Omega}(\frac{d^{5/3}}{m^{1/3}})$ for randomized and deterministic algorithms respectively. Notably, due to existing upper bounds, our lower bound for deterministic algorithms is the first to show a sharp oracle complexity phase transition around $m\approx d^2$, where a polylogarithmic change in memory leads to a $\mathsf{poly}(d)$ change in the number of required oracle calls. Further, when the suboptimality is polynomially small in $d$, our lower bound randomized algorithms is the first to show that $\tilde{\Omega}(d^2)$ memory is necessary to nearly match the optimal query complexity among algorithms without memory constraints. Previously, such a result was only known for the regime where the suboptimality is quasipolynomially small in $d$.

\end{abstract}

\section{Introduction}
Convex optimization is a fundamental problem which has been studied through the oracle complexity model as far back as \cite{nemirovsky85}. In the first order oracle model, an algorithm must minimize an unknown function $F:\cB(1)\mapsto \re$ (here $\cB(1)$ denotes the unit ball) 
via access through a first order oracle, $\cO$, which at each query $w \in \cB(1)$ returns $F(w)$ and $g \in \partial F(w)$, i.e. an element in the subgradient of $F(w)$\footnote{Lower bounds for the case where the algorithm can make queries outside the ball, or even output a solution outside the ball (i.e. the unconstrained setting) can be obtained by a generic reduction, \cite[Lemma 44]{MSSV24-Arxiv}.}. 
Under the guarantee that $F$ $1$-Lipschitz and convex, 
the optimizer attempts to find a point $\widehat{w}$ satisfying $F(\widehat{w})-\min_{w\in\cB(1)}\{F(w)\} \leq \alpha$ in the fewest number of oracle calls. 

Without memory constraints, this problem is well understood, and the oracle complexity is known to be $\Theta(\min\{1/\alpha^2,d\log(1/\alpha)\}$ \cite{nemirovsky-information-based-complexity, BGP17}. When $\alpha \geq 1/\sqrt{d}$ (or equivalently, $d\geq 1/\alpha^2$), often referred to as the large scale regime, gradient descent achieves the optimal upper bound. Even in the high accuracy regime $\alpha < 1/\sqrt{d}$, gradient descent remains the dominant method in practice. Cutting plane methods such as the center of gravity method achieve the optimal $O(d\log(1/\alpha))$ query complexity in the high accuracy regime, but come with significant drawbacks. Notably, they all require at least $\Omega(d^3\log(1/\alpha))$ running time due to the complexity of computing the oracle query points \cite{LSW15,jiang20-improve-cutting-plane}\footnote{It should be noted that the methods which we can run in $\tilde{O}(d^3)$ time are more involved than center of gravity.}.
They also all require $\tilde{\Omega}(d^2)$ bits of memory to store the cutting planes. This contrasts sharply with the $O(d\log(1/\alpha))$ bits of memory needed for gradient descent. Our intuition for what makes optimization challenging in the high accuracy regime is substantially limited by such gaps.

This motivates the COLT 2019 open problem of characterizing memory-query tradeoffs in convex optimization \cite{woodworth19-opt-with-limited-memory}. First progress on this problem was made by \cite{MSSV22} for randomized algorithms, who showed that optimizing to $\alpha = 1/\mathsf{poly}(d)$ accuracy with $m$ bits of memory requires $\tilde{\Omega}\big(\frac{d^{8/3}}{m^{4/3}}\big)$ oracle queries. 
In the regime $m\geq d^{8/7}$ and $\alpha$ quasipolynomially small in $d$ (specifically, $\alpha \leq 2^{-\log^5(d)}$), this was improved by \cite{CP23}, who showed $\tilde{\Omega}\big(\frac{d^{4/3}}{m^{1/6}}\big)$ queries are necessary. For instance, this implies that $m = d^{2-\delta}$ bits of memory require $\tilde{\Omega}(d^{1+\delta/6})$ queries. 
These are the best previously known lower bounds for randomized algorithms, but \cite{BZJ23} proved a stronger result for deterministic algorithms, and showed that $\tilde{\Omega}\big(\frac{d^{5/3}}{m^{1/3}}\big)$ queries are required. Equivalently, $m=d^{2-\delta}$, implies $\tilde{\Omega}(d^{1+\delta/3})$ queries are required.

\subsection{Contributions.}
Our contributions are the following. We first introduce a new primitive, the \emph{marked subspace game with hint} (MSGH), which we use as a tool to analyze the memory constrained optimization problem.
Notably, this game provides fundamental insight into the algorithmic limitations of making queries orthogonal to a high entropy random matrix when given a limited amount of informative bits. 
Using the MSGH, we prove new lower bounds for memory constrained optimization.
This proof structure is in keeping with past work, although our game and proof technique differ substantially from those previously developed.
Using the MSGH, we show a new oracle complexity lower bound for randomized algorithms with $m$ bits of memory. Such algorithms must make $\tilde{\Omega}\big(\frac{d^2}{\sqrt{m}}\big)$ queries to find a $1/\mathsf{poly}(d)$ accurate solution.
Notably, this shows that near quadratic memory is necessary to nearly match the optimal $\tilde{O}(d)$ oracle complexity among algorithms without memory constraints, even for accuracy polynomially small in $d$. Previous work had shown such a result only for the case when $\alpha \leq 2^{-\log^5(d)}$ \cite{CP23}.
For deterministic algorithms, we provide an even stronger lower bound of $\tilde{\Omega}(\min\{d^{1.6},\frac{d^{8/3}}{m^{2/3}}\})$. 
This bound is the first to show a sharp transition around $m=d^2$. Specifically, for some $m=O(d^2\log(1/\alpha))$, Vaidya's algorithm achieves oracle complexity $O(d\log(1/\alpha))$ \cite{blanchard23-mem-constrained-algorithms}. But for some $m = \Omega(\frac{d^2}{\log(d)})$, our results show the necessary query complexity jumps by a $\mathsf{poly}(d)$ factor to $\tilde{\Omega}(d^{4/3})$. 
Thus, \emph{any} optimization algorithm in the oracle model that improves on the memory complexity of Vaidya's algorithm, \emph{even by a poly-logarithmic factor}, must suffer polynomial loss in the oracle complexity. By contrast, previous lower bounds did not show such a sharp transition, and instead only showed oracle complexity degrades by a polynomial factor if $m=O(d^{2-\delta})$ with $\delta=\Omega(1)$.

\subsection{Comparison to Previous Approaches.} 
\paragraph{Overview of previous lower bounds.}
Three previous works, \cite{MSSV22, BZJ23, CP23}, have provided lower bounds for the memory constrained optimization problem. They all analyze a hard loss function instance that roughly takes the form,
$F(w) = \max\{ \max_{j\in[N]}\{\ip{x_j}{w} - j\gamma\}, \|Aw\|_\infty - \rho\},$
where the first term in the max is a ``Nemirovski" function defined using $N$ randomly sampled vectors, $x_{1}, \ldots, x_{N}$, and offset $\gamma>0$, and the second term is a ``barrier'' function which uses a high entropy matrix $A$ and offset $\rho>0$. Query complexity lower bounds are proven by showing that with low memory, the optimizer cannot memorize $A$, and thus needs $\tilde{\Omega}(d)$ queries to ``discover'' the next few Nemirovski vectors.
In both \cite{MSSV22} and \cite{BZJ23}, a running time lower bound is proven by showing the optimization algorithm repeatedly solves the ``orthogonal vector game'' (OVG), where the goal of the game player is to find $k$ well spread out vectors nearly orthogonal to $A$.
If the memory of the optimizer is $m=\Theta(kd)$, this allows them to show that discovering the next $k$ Nemirovski vectors requires $\Omega(d)$ oracle queries. They thus achieve query lower bounds of $\Omega\big(\frac{dN}{k}\big)$, where $N$ is limited by various factors, most notably the needed scale for the offset, $\gamma$. \cite{MSSV22} was able to set $N=\tilde{\Omega}\big(\frac{d}{k}\big)^{1/3}$, achieving a $\tilde{\Omega}\big(\frac{d^{4/3}}{k^{4/3}}\big)$ query lower bound for randomized algorithms. 
Subsequently, \cite{BZJ23} leveraged determinism and a modification of the OVG (the OVG with hints) to increase $N$ and get a $\tilde{\Omega}\big(\frac{d^{4/3}}{k^{1/3}}\big)$ lower bound for deterministic algorithms. 

The work \cite{CP23} introduced a new game, the orthogonal \textit{correlated} vector game (OCVG), which requires the algorithm to find a single vector orthogonal to $A$, but under the additional requirement that the vector correlates with a randomly sampled vector $x$. Again analyzing a similar loss construction, they use this new game to show that roughly $d$ queries are needed to discover just one new Nemirovski vector. 
This gives a $\tilde{\Omega}(Nd)$ lower bound, 
but also requires a larger offset, and consequently uses $N=o((d/k)^{1/3})$.
The final result is a 
$\tilde{\Omega}\big(\frac{d^{7/6}}{k^{1/6}}\big)$ query lower bound for randomized algorithms.

\vspace{-5pt}\paragraph{Our approach.} 
An intuition underlying the analysis in these previous lower bounds is that it is hard for the optimizer to ``explore'' a large subspace orthogonal to $A$ because of its limited memory. For example, consider the strategy where the optimizer uses its $kd$ bits of memory to store a $\tilde{\Theta}(k)$-dimensional linear subspace orthogonal to $A$. Intuitively, one would expect that the only queries the optimizer can make that avoid the barrier term in the loss are in this small subspace.

In our work, we improve upon existing analysis of the loss construction by introducing the Marked Subspace Game with Hint (MSGH). Roughly, this game shows that the behavior of any information limited strategy for querying points nearly orthogonal to $A$ mimics the simple strategy just described.
In the game, the player first gets to choose and store a large message about $A$. Then, an adversary is allowed to ``mark'' a $O(k)$-dimensional linear subspace based on the player's message. The player only wins the game if they find a point far from the marked subspace that is nearly orthogonal to $A$. 
We show the MSGH is not winnable with a $dk$-bit message about $A$, even if the player receives a small ``hint'' after a subspace is marked. 

Using this new game, we improve lower bounds for both randomized and deterministic memory-constrained optimizers. Specifically, the MSGH will show that at any point during the optimization procedure, any query the optimizer makes that avoids the barrier function lies near a $k$-dimensional subspace (i.e. the marked subspace).
For deterministic algorithms, we leverage the fact that the Nemirovski vectors can be sampled adaptively to the optimization procedure. When the optimizer reaches a point where it can achieve non-trivial correlation with the current Nemirovski vector, we sample the next Nemirovski vector orthogonal to all previously ``marked'' subspaces. 
By sampling the vector adaptively in this way, we can use a very small value for the offsets, and thus embed more vectors. 
A bottleneck here is that the past Nemirovski vectors may leak a non-trivial amount of information about $A$ while the optimizer is searching for the next vector. Carefully controlling this leakage is an important part of our analysis. This issue was also tackled by \cite{BZJ23}, although our techniques in this regard differ.
Overall, we are able to embed roughly $\min\{d^{3/5},\frac{d}{k^{2/3}}\}$ vectors, giving a lower bound of $\tilde{\Omega}\{\min\{d^{1.6},\frac{d^2}{k^{2/3}}\}\}$.

For randomized algorithms, we observe that the MSGH gives a tight lower bound for the orthogonal correlated vector game from \cite{CP23} (i.e. a strategy for winning the OCVG will win the MSGH, under appropriate parameters). Specifically, we show the OCVG is not winnable when the correlation requirement for the randomly sampled vector $x$ is roughly $\sqrt{k/d}$, whereas $\cite{CP23}$ required a threshold of $(k/d)^{1/4}$.
By itself, this would allow us to fully remove the $1/k$ factor from the \cite{MSSV22} lower bound while preserving $N=(d/k)^{1/3}$. However, we push our bound here even farther by introducing a modified version of the wall function from \cite{bubeck19_highlyparallel} into 
the loss construction.
The idea behind this addition is that because the OCVG cares about degree of correlation with the target vector, optimization progress is harder when the algorithm is forced to make queries with small norm, which is enforced by the wall function. This allows us to decrease $\gamma$ further and increase $N$ to $N=\sqrt{d/k}$, achieving a $\tilde{\Omega}(Nd) = \tilde{\Omega}\big(\frac{d^{1.5}}{\sqrt{k}}\big)$ lower bound. Notably, this is stronger even than the previously best known lower bound for deterministic algorithms \cite{BZJ23}.

\subsection{Additional Related Work.}
Outside of first order optimization, many other machine learning and statistical estimation problems have been studied under memory constraints, such as online learning, linear regression, detecting correlations, and statistical queries, to name just a few examples \cite{SH14,dagan18a,DKS19,Steinhardt15, SteinhardtVW16,moshkovitz_et_al_18,peng23}. Memory constrained learning is also closely related to learning under communication or information constraints, which has been a topic of increasing interest \cite{braverman16, bassily18-little-information, brown2021memorization, attias2024infocomplexitySCO}. 
Another closely related problem is the feasibility problem, which several works have studied under memory constraints \cite{BZJ23, blanchard23-mem-constrained-algorithms, blanchard24}. This problem is related in that lower bounds for first order optimization yield lower bounds for feasibility, although the converse does not hold. 
In addition to lower bounds, memory efficient algorithms have also been studied for convex optimization \cite{nocedal89, Liu1989OnTL, pilcanci17-newtonsketch}, although the only one to provide formal guarantees which surpass standard cutting plane methods is the recursive cutting plane method of \cite{blanchard23-mem-constrained-algorithms}. Unfortunately, this method only gives significant improvements for the regime where the suboptimality is  exponentially small in $d$.

Outside of memory constraints, oracle complexity lower bounds have been studied for decades \cite{nemirovsky85,nemirovsky-information-based-complexity}. Oracle complexity bounds not only give insight into the running time limitations of optimization, but also provide intuition for algorithm design. Traditionally such bounds have been expanded by exploring different regularity conditions on the loss, but more recent work has explored oracle complexity under algorithmic constraints (outside of limited memory) such as limited adaptivity, intermittent communication, gradient compression, and privacy \cite{arjevani2015communication,DG20, woodworth21a, acharya-info-constrained-opt, MN25}.

\section{Preliminaries}
\paragraph{Memory constrained convex optimization.}
We consider the class of optimization algorithms which make at most $n$ calls to a first order oracle and have access to $m$ bits of memory, which we denote as $\cM_{\mathsf{rand}}(m,n)$ and $\cM_{\mathsf{det}}(m,n)$ for randomized and deterministic algorithms respectively. The algorithm may use arbitrary memory in between oracle queries, but we require that the algorithm's state be representable by $m$ bits at some point between each query to the first order oracle. The first order oracle for the loss/objective $F$ can be any function $\cO:\cB(1) \mapsto \re\times \re^d$ (where $\cB(1)$ is the unit Euclidean ball) such that $\forall w\in\cB(1): \cO(w) = (F(w),g_w)$ for some $g_w \in \partial F(w)$, with $\partial F$ denoting the subgradient. We are interested in the ability of such optimization algorithms to find approximate minima of a $1$-Lipschitz convex function $F$ under the unit ball constraint, where the suboptimality of $\widehat{w}$ is measured by $F(\widehat{w})-\min_{w\in\cB(1)}\{F(w)\}$.

\vspace{-5pt}\paragraph{Notation.} We use $\cB(r)$ to denote the Euclidean $d$-dimensional ball of radius $r$; $\cS(r)$ will denote the $(d-1)$-dimensional sphere. The Minkowski sum of two sets $\cA$ and $\cB$ is denoted with $\cA\msum\cB$.
$\Span(v_1,..,.v_n)$ denotes the span of a set of vectors and $\Span^\perp(...)$ denotes the orthogonal complement of the span.  $\mathsf{Gr}(k,d)$ denotes the Grassmannian which is the set of all $k$-dimensional linear subspaces of $\re^d$. We let $\Gr([k],d) = \bigcup_{j\in[k]} \Gr(j,d)$. $\proj_S$ denotes the orthogonal projection onto $\Span(S)$ and $\proj_S^\perp$ denotes projection onto its orthogonal complement. $\proj_{S,S'}$ denotes orthogonal projection onto $\Span(S\cup S')$. Collections of vectors $u_1,...,u_t$ may be indexed via $u_{< l} = \{u_q\}_{q\in[l-1]}$, $u_{p:l}=\{u_p,...,u_l\}$, or $u_{\neq l} = \{u_1, \ldots, u_{l-1},u_{l+1}, \ldots,u_t\}$.
For $p\in\mathbb{Z}^+$, we use $[p]$ to denote the set $\{1,\dots,p\}$.
Finally, we note that several standard lemmas on random projections and packings/covers will be used frequently throughout. We defer these to Appendix \ref{app:lemmas}.

\section{Marked Subspace Game} \label{sec:game}
Algorithm \ref{alg:MSGH} details the new primitive we consider for our analysis, the Marked Subspace Game with Hint (MSGH). Here, the player first stores some large message about the random matrix $A$. Then, depending on the message and $A$, a $k$-dimensional subspace is ``marked''.  The player then receives a smaller ``hint'' message, which can be adaptive to the marked subspace.
Finally, the player is allowed to make $T$ row queries, and wins if any such query is nearly orthogonal to $A$ and far from the marked subspace.
Note that the player can still make row queries in the marked subspace, but such queries do not on their own result in the player winning the game.
Also, the player can implement any randomized version of the message/hint by randomly sampling $h_1$ and $h_2$. Since $\cL$ is chosen based on $h_1$, the adversary functionally has access to (at least part of) the random seed of $\cP$ when picking $\cL$. 
We have the following statement about the difficulty of the game.

\begin{algorithm}[t]
\caption{\textit{Marked Subspace Game with Hint (MSGH)}}
\label{alg:MSGH}
\begin{algorithmic}[1]
\REQUIRE embedding dimension $d$, matrix size $d'\in[d]$, query limit: $T \in [d']$, message size $m_1 \in [d^2]$, hint size $m_2 \in \{0,...,d\}$, dimension threshold $k > 0$, row norm parameters $\Delta_1,...,\Delta_{d'}$
\STATE Player $\cP$ chooses $h_1:\re^{d'\times d} \mapsto \{0,1\}^{m_1}$ and $h_2: \Gr([k], d) \times \re^{d'\times d} \mapsto \bit^{m_2}$ 
\STATE Sample $A \in \re^{d'\times d}$ s.t. for each $i\in [d']$, row $i$ is sampled as $a_i \sim \Unif(\cS(\Delta_i))$
\STATE Adversary picks (``marks'') a linear subspace $\cL \in \Gr([k],d)$ using $A$ and $h_1$  
\STATE $\cP$ receives $B=h_1(A)$ and $q = h_2(\cL,A)$ \label{line:player-info}
\FOR{$t=1...T$}
    \STATE $\cP$ queries $u_t \in \re^d$ and $\eta_t \in \re^{d'}$.
    \STATE $\cP$ receives $\hat{a}_t = a_{i_t}$ where $i_t= \argmax\limits_{i \in [d']}\bc{\ip{u_t}{a_i} - \eta_{t,i}}$
\ENDFOR
\STATE $\cP$ wins if ~$\exists t\in[T]$ satisfying both 
\begin{align*}
    &\text{1) Near orthogonality: } ~~\|A u_t\|_\infty \leq \frac{\|u_t\|}{2 d^{4.5}}, \quad \text{2) Subspace distance: } \frac{u_t}{\|u_t\|} \notin \cL \msum \cB\big(d^{-1.5}\big)
\end{align*}
\end{algorithmic}
\end{algorithm}
\begin{proposition}\label{prop:msgh-bound}
Let $C_1, C_2 > 1$ be universal constants and $d$ larger than some constant. Let $d' \geq d/2$,
$T < \frac{d}{C_1\log(d)}$, 
$m_1 \leq \frac{d^2}{64\log(d)}$, $m_2 \leq \frac{d}{C_2}$ and $k\geq 32 \max\{\frac{m_1}{d}, \log(d)\}$.
There exists a strategy for the adversary such that any (potentially randomized) player loses the MSGH w.p. at least $1-O(\frac{1}{d^3})$. 
\end{proposition}
We provide some discussion before providing the proof in Section \ref{sec:msgh-proof}.
Another way of looking at this result is that for any query strategy, there exists some subspace $\cL$ such that any query which is nearly orthogonal to $A$ lies near this $k$-dimensional subspace. This seems to closely mimic the natural upper bound strategy where the player uses $m_1$ to store an arbitrary $\tilde{\Theta}(k)$-dimensional subspace orthogonal to $A$. 
One would expect the conditional distribution of $A$ to be (roughly) uniformly random in the space orthogonal to the space stored by $m_1$. As such, we would indeed infer that any attempt at a query outside the stored space will likely not be orthogonal to $A$ due to random projection. 

In our lower bound for deterministic optimization algorithms, we will sample the Nemirovski vectors adaptively and orthogonally to $\cL$. This means they may leak additional information about $A$, which is why the MSGH accounts for a hint $q$. 

\paragraph{Orthogonal Correlated Vector Game (OCVG).} 
We can use the MSGH to provide a tighter analysis for the Orthogonal Vector Correlated Game (OCVG) introduced in \cite{CP23} (see also \cite[Remark 2.2]{CP23}). 
Specifically, we can lower the correlation threshold from $(k/d)^{1/4}$ to $\sqrt{k/d}$ under the same memory requirement. 
\begin{proposition}\label{prop:ocvg-bound}
Consider the modification of Algorithm \ref{alg:MSGH} where $\cP$ also receives $x\sim \Unif(\cS(1))$ at Line \ref{line:player-info}. 
Then under the same parameter regimes specified in Proposition \ref{prop:msgh-bound},
with probability at least $1-O(1/d^3)$ there does not exist any $t\in[T]$ such that $u_t$ satisfies both the modified win conditions,
\begin{align*}
  \text{1) Near orthogonality: } ~~\|A u_t\|_\infty \leq \frac{\|u_t\|}{2 d^{4.5}} \quad  \text{2) Sufficient correlation: }|\ipnos{u_t}{x}| \geq 3\|u_t\|\sqrt{\frac{k}{d}}  
\end{align*}
We refer to this modification as the orthogonal correlated vector game (OCVG) \footnote{In minor ways, our OCVG does not match exactly the game defined in \cite{CP23}. For example, our win condition depends on the norm of the query.}.
\end{proposition}
\begin{proof}%
Observe that the subspace chosen by the adversary only depends on $A$ and $h_1$, and is thus independent of $x$.
Thus by standard concentration results for random projections (see Lemma \ref{lem:random-projection}), we have for any choice of $\cL$ such that $\Dim(\cL) \leq k$ that
$\pr{x}{\max\limits_{v\in\cL\cap\cS(1)}\{|\ip{x}{v}|\}  \geq 2\sqrt{\frac{k}{d}}} \leq 2e^{-k/2} \leq \frac{1}{d^3}$ (recall $k> 32\log(d)$).
Assuming $\max\limits_{v\in\cL\cap\cS(1)}\{|\ip{x}{v}|\}$  is bounded by  $2\sqrt{\frac{k}{d}}$, we have that for any unit vector $v\in \cL\msum \cB(1/d^{1.5})$, $|\ip{x}{v}| \leq 2\sqrt{k/d} + d^{-1.5} \leq 3\sqrt{k/d}$.
On the other hand, we know by Proposition \ref{prop:msgh-bound} that there exists a choice of $\cL$ (made independently of $x$) such that w.p. at least $1-O(1/d^3)$ only queries in $\cL \msum \cB(d^{-1.5})$ satisfy near orthogonality. We have just shown that with probability at least $1-O(1/d^3)$ no such query can also satisfy sufficient correlation, so the proof is complete.
\end{proof}
Note that while we have used the marked subspace as a tool in the above proof, the OCVG game itself essentially ignores the marked subspace, particularly if $m_2=0$. %
The OCVG lower bound is what we will use in our analysis for randomized algorithms.
The correlation threshold is near tight, as it is matched by the strategy which stores a $\tilde\Omega(k)$ dimensional subspace orthogonal $A$, and queries the projection of $x$ onto the stored subspace.
We also note that the flexibility added to the game via the row norm parameters and query offsets $\eta_{1:p}$ is needed for technical reasons related to handling the wall function (see Section \ref{sec:randomized}). The hint $q$ is not needed for this application, but does not hurt asymptotics.

\subsection{Proof of Proposition \ref{prop:msgh-bound}} \label{sec:msgh-proof}
Since a randomized algorithm can be described as a distribution over deterministic algorithms, it suffices to show that the proposition holds for any deterministic algorithm (i.e. Yao's minimax principle). Note on this point that if $\cP$ were to implement a randomized mapping from $A$ to the bits (which it can do by randomly selecting $h_1$), it must still give the adversary the deterministic function $h_1$. In this way, the adversary functionally has access to the random seed of $\cP$ used for this first part of the game.
Further, without loss of generality we can assume for all $t\in[T]$ that $\|u_t\|=1$, since this does not change whether $u_t$ wins the OCVG or the row returned from $A$.

\paragraph{Proof preliminaries. } 
Let $\tilde{A}$ be the matrix of rows of $A$ not returned by any query made during the MSGH. 
The following set will be crucial in our analysis. For a constant $c_1$, define, %
\begin{align*}
 \cG_b = \Big\{u\in\cS(1): \bigpr{A|B=b}{\|\tilde{A} u\|_\infty \leq \frac{1}{d^{4.5}}} \geq 2^{-d/c_1} \Big\},   \quad \forall b\in\bit^{m_1}.
\end{align*}
We recall $\tilde{A}$ are the rows of $A$ not returned by the query responses.
Intuitively, $\cG_B$ is the set of points which $\cP$ is able to identify as orthogonal to $\tilde{A}$ with non-trivial probability (note the prior probability of any $x\in\cS(1)$ being near orthogonal to $\tilde{A}$ is $2^{-\Omega(d)}$). We informally refer to $\cG_B$ as the ``strong posterior set''. Our goal in the rest of the proof is to show that, roughly, the strong posterior set can contain only a $k$-dimensional subspace of points which satisfy the near orthogonality condition. Thus by marking this subspace, the adversary forces the player to win via some query in the complement of $\cG_B$, the ``weak posterior set''. However, by definition, $\cP$ has little information about which points in the weak posterior set win. The final step is showing that even after additional information provided by the hint and row queries, the weak posterior set is still weak.

It will be helpful to introduce the ``slice-dimension''. In words, the slice dimension of a set is the lowest dimension of a linear subspace such that all points in the set are close it. %

\begin{definition}\label{def:slice-dim}
The $\tau$-slice dimension of a set $\cY\subseteq \re^d$, $\Gamma(\cY,\tau)$, is defined as,
$\Gamma(\cY,\tau) = \min_{\cL}\{\Dim(\cL): \cY\subseteq \cL\msum \cB(\tau)\},$
where the minimum is taken over all linear subspaces of $\re^d$.
\end{definition}
 
\paragraph{Choosing the marked subspace.}
Our aim is now to show that there is a good choice of the marked subspace, $\cL$, which prevents $\cP$ from succeeding by guessing a point in $\cG_B$. 
\begin{lemma} \label{lem:ortho-set-small-width}
For any $A'\in\Supp(A)$ and $b\in \bit^m$, let 
$\cY_{b,A'} := \{v\in\cG_b: \|A'v\|_\infty \leq \frac{1}{d^{4.5}}\}$. Under the parameter settings in Proposition \ref{prop:msgh-bound}, it holds that
$\pr{A}{\Gamma(\cY_{B,A}, \frac{1}{2d^{1.5}}) \leq k} \geq 1-2^{-\Omega(dk)}$.
\end{lemma}
\vspace{-5pt}
Given the lemma, the strategy for the adversary is to mark the linear subspace $\cL$ for which $\cY_{A,B} \in \cL \msum \cB(d^{-1.5})$, and whose existence is guaranteed to exist by the slice dimension bound.
The proof of the lemma is deferred to Appendix \ref{app:ortho-set-small-width}, and the rest of this paragraph section is devoted to highlighting the key ideas and auxiliary lemmas needed to establish Lemma \ref{lem:ortho-set-small-width}. %

The high-level idea is to show that if the slice dimension were large, then one could use the function $h_1$ to encode a linear subspace, such that with non-negligible probability \textit{every} point in the subspace is  orthogonal to $A$ (note $\cG_B$ only cares about a point's marginal probability of being orthogonal). We can then show that this is not possible via the following lemma.
\begin{lemma}\label{lem:no-large-ortho-space}
Let $d'\geq d/2$, $k\in[d]$, and let $A \in \re^{d'\times d}$ 
have rows sampled as uniformly random vectors of norms $\Delta_1,...,\Delta_{d'} \in [0.5,1]$. 
Let $g:\re^{d'\times d} \mapsto \Gr([d], d)$
such that $|\mathsf{Range}(g)| \leq 2^{kd'/16}$. Then, 
$\PP_A\Big[\Dim(g(A)) \geq k ~\land~ \max\limits_{u\in g(A)\cap \cS(1)}\{\|A u\|_\infty\} \leq \sqrt{\frac{k}{16d}}\Big] \leq 2e^{-d'k/16}.$
\end{lemma}
Note the range condition follows when $g$ is computed using $kd'/16$ bits that depend on $A$. 
The proof is deferred to Appendix \ref{app:no-large-ortho-space}, and uses the fact that any fixed linear subspace of dimension $k$ is nearly orthogonal to $A$ with probability at most $2^{-\Omega(d'k)}$ by the properties of random projection. Then taking a union bound over all possible outputs of $g$ with large dimension shows that w.h.p. there is no subspace in the range of $g$ that satisfies the stated condition.

Lemma \ref{lem:no-large-ortho-space} is not sufficient in itself to establish Lemma \ref{lem:ortho-set-small-width}. Even if there is a large set of points in $\cG_B$ which are orthogonal to $A$, identifying which points these are requires additional access to $A$.
To bridge this gap, we have the following lemma, which shows that if a set of points nearly orthogonal to $A$ has large slice dimension, then one can find a linear subspace which is entirely orthogonal to $A$ by specifying some number of vectors from the set. 
\begin{lemma} \label{lem:slice-dim-yields-ortho-space}
Let $\gamma > 0$, $\hat{d}\in[d]$ and $\cG\subseteq \re^d$.
Then there exists a function $f_{\cG}:\re^{d'\times d} \mapsto \Gr(\hat{d},d)$ such that $\log(|\Range(f_{\cG})|) \leq \hat{d}\log(|\cG|)$ and
whenever there exists $\cY\subseteq \cG$ such that 
$\forall v\in\cY: \|Av\|_\infty \leq \gamma$, and
$\Gamma(\cY,\tau)\geq\hat{d}$, then
$\max_{v\in f_\cG(A)}\{\|Av\|_\infty\} \leq \|v\|\frac{\gamma d}{\tau}.$
\end{lemma}
We defer the proof to Appendix \ref{app:slice-dim-yields-ortho-space}. Roughly, the proof proceeds by choosing $\hat{d}$ vectors from $\cY$ that span a parallelepiped of maximal volume, and having $f(A)$ be the span of these vectors. We use the large slice dimension assumption to guarantee that the selected vectors form a ``good basis'', i.e., each of them has a substantial component orthogonal to the other vectors. Via a linear algebraic argument, this then guarantees that for any vector $v$ in their span that $\|Av\|_\infty$ is appropriately bounded.

Using the above lemma, if the slice dimension of $\cY_{B,A}$ is $k$, we can map to a $k$-dimensional linear subspace $\cL$ such that every point in the space is nearly orthogonal to $A$. This process uses $m_1 +k\log(|\cG_B|)$ bits, where $m_1$ bits are used to specify $\cG_B$, and the rest are used to select appropriate vectors in $\cG_B$ which determine the subspace. 
The final piece of the proof is showing that $\cG_B$ (or rather, an appropriate packing-cover of $\cG_B$) is not too big, so that the number of bits needed to specify $\cL$ can be appropriately bounded and we can derive a contradiction from Lemma~\ref{lem:no-large-ortho-space}.
Proving the bound on (the packing-cover of) $\cG_B$ uses the machinery developed in Lemmas \ref{lem:no-large-ortho-space} and \ref{lem:slice-dim-yields-ortho-space}. The main difference is that instead of explicitly considering the case where the slice dimension is large, we use the fact that if the cover is very large, then with non-negligible probability many points in it will be nearly orthogonal to $A$, which follows from the way $\cG_B$ is defined. Then, we can use the fact that $\cG_B$ is also a packing to bound its slice dimension from below.

\paragraph{Bounding the impact of additional information. }
The preceding section established that w.h.p. the marked subspace can be chosen such that there are no winning points in the stronger posterior set $\cG_B$. 
We introduce two more lemmas which we will use to bound the change in the weak posterior set after $\cP$ receives additional information.
The hint can be handled nicely with the following lemma, which shows that few bits do not update the posterior much. 
The proof is in Appendix \ref{app:bits-bound-posterior}.
\begin{lemma}\label{lem:bits-bound-posterior}
Let $A,B$ be random variables where $\Supp(B)=\bit^s$, $s\in\cZ^+$. Then for any $\delta \in [0,1]$,
$\pr{B}{\exists \cE \subseteq \Supp(A): \pr{}{A\in\cE | B} \leq 2^{s+\log(1/\delta)}\pr{}{A\in\cE}} \leq \delta.$
\end{lemma}

Finally, we need to handle the fact that the row queries may leak information about $\tilde{A}$. For example, observe that on receiving row $a_{i_t}$ on query $(u_t,\eta_t)$, $\cP$ at least learns $|\ip{a_j}{u_t}|-\eta_{t,j} \leq |\ip{a_{i_t}}{u_t}|-\eta_{t,i_t}$ for all $j\neq i_t$. The proof of the following lemma can be found in Appendix \ref{app:G-updates}. %
\begin{lemma} \label{lem:G-updates}
Let $\cG(\beta, t) = \{u\in\cS(1): \PP[\|\tilde{A} u\|_\infty \leq \frac{1}{d^{4.5}}] \geq \beta \}$, where the probability is taken with respect to the conditional distribution of $A$ given the realizations of $B$, $q$, $u_{\leq t}$, and $\hat{a}_{\leq t}$. Then for any $\beta_0\in[0,1]$, 
with probability at least $1-\frac{1}{d^3}$ (w.r.t. $A$) it holds $\forall t\in [T]$ that
$\cG(2^{4t\log(d)}\beta_0, t) \subseteq \cG(\beta_0, 0).$
\end{lemma}

\paragraph{Concluding the MSGH proof. } We can now prove the difficulty of winning the MSGH. 
\begin{proof}[Proof of Proposition \ref{prop:msgh-bound}]
By Lemma \ref{lem:ortho-set-small-width}, we have 
$\pr{A}{\Gamma(\cY_{B,A}, \frac{1}{d^{1.5}}) \leq k} \geq 1-2^{-\Omega(dk)}$. Conditional on this event, the adversary picks the linear subspace, $\cL$, such that $\cY_{B,A} \subset \cL \msum \cB(d^{-1.5})$, which exists by the definition of the slice dimension. Since $\cY_{B,A}$ is the set of points in $\cG_B$ which satisfy near orthogonality, this implies that no point in $\cG_B$ wins the MSGH. What remains is to show $\cP$ cannot win via a point in $\cS(1) \setminus \cG_B$. 
Recall $u\in\cG_B$ if $\bigpr{A|B=b}{\|\tilde{A} u\|_\infty \leq \frac{1}{d^{4.5}}} \geq 2^{-d/c_1}$.
Now assuming $C_2$ is large enough such that $m_2 \leq \frac{d}{3c_1}$,
Lemma \ref{lem:bits-bound-posterior} implies that for any $\delta \in[0,1]$, with probability at least $1-\delta$ under the draw of $q$ we have,
$\forall v\in \cS(1): \PP_{A|B,q}[\|\tilde{A} v\|_\infty \leq \frac{1}{d^{4.5}}] \leq 2^{\frac{d}{3c_1} + \log(1/\delta)}\PP_{A|B}[\|\tilde{A} v\|_\infty \leq \frac{1}{d^{4.5}}].$
Setting $\delta = 2^{-\frac{d}{3c_1}}$ we get with high probability
that $\forall v\in \cG_{B}: \pr{A|B,q}{\|\tilde{A} v\|_\infty \leq \frac{1}{d^{4.5}}} \leq 2^{\frac{2d}{3c_1} - \frac{d}{c_1}} \leq 2^{-\frac{d}{3c_1}}$. Thus
$\cG(2^{-\frac{d}{3c_1}}, 0) \subseteq \cG_{B}$.

We now apply Lemma \ref{lem:G-updates}, which implies 
$\cG(2^{{4t\log(d)-\frac{d}{3c_1}}}, t) \subseteq \cG(2^{-\frac{d}{3c_1}},0) \subseteq \cG_B$,
with probability at least $1-O(1/d^3)$. Under this event, for any $t\in[T]$, $u_t \notin \cG_B \implies u_t \notin \cG(2^{{4t\log(d)-\frac{d}{3c_1}}}, t)$ and thus by the definition of $\cG(2^{{4t\log(d)-\frac{d}{3c_1}}},t)$, $u_t$ does not satisfy near orthogonality with probability more than $2^{{4t\log(d)-\frac{d}{3c_1}}} = 2^{-\Omega(d)}$ provided $C_1$ is large enough. Taking a union bound over all queries, we obtain w.p. at least $1-O(1/d^3) - d 2^{-\Omega(d)} = 1-O(1/d^3)$ that $\cP$ does not win the MSGH. 
\end{proof}

\section{Lower Bound for Deterministic Algorithms} \label{sec:deterministic}

For deterministic optimizers, we give the following lower bound. We ignore the case where $m = o(d\log(d))$ since in this regime the problem is not solvable with any query budget \cite[Theorem 5]{woodworth19-opt-with-limited-memory}.
\begin{theorem}  \label{thm:detalg-lb}
Let $C_3$ be a universal constant, and assume $d$ is larger than some constant.
Let
$m \in [d\log(d), \frac{d^2}{128\log(d)}]$,
$n \leq \frac{1}{C_3\log(d)}\min\{d^{1.6}, \frac{d^{8/3}}{m^{2/3}}\}$
and $\opt \in \cM_{\mathsf{det}}(m,n)$.
Then there exists a $1$-Lipschitz convex loss function such that $\opt$ does not output a solution, $\widehat{w}\in\cB(1)$, satisfying
$F(\widehat{w})-\min_{w\in\cB(1)}\{F(w)\} \leq \frac{1}{4d^{6.5}}$.
\end{theorem}
We present the proof subsequently. The first step is to construct the hard objective adaptively. We will do this using a ``resisting oracle'' which will only reveal information about certain pieces of the loss function after many queries have been made. The main crux of the proof is then showing that this resisting oracle gives the same responses the true oracle would. We do this by showing that to discover new pieces of the loss function, the optimizer must essentially play the MSGH. Finally, if the optimizer does not make sufficient progress on the objective function, we show this leads to a suboptimal solution.

\subsection{Proof of Theorem \ref{thm:detalg-lb}}
\paragraph{Hard instance. }
By standard reductions, it suffices to show a lower bound for algorithms which achieve $\alpha=\frac{1}{\sqrt{d}}$ accuracy on $4d^{6}$-Lipschitz loss functions.
As we consider deterministic algorithms, the exact loss construction will depend on the queries made by $\opt$.
We first give only a general structure for the family of losses we consider.
In the following, $C_1,C_2$ are as in Proposition \ref{prop:msgh-bound}.
Further, we assume that $m \geq d^{8/5}$. Our result straightforwardly extends to lower memory regimes, but does not get stronger as $m$ decreases below $d^{8/5}$.
Let $k= \frac{32m}{d}$. Note $k \geq 32d^{3/5}$ by the assumption on $m$.
Let $\gamma_1 = 11\sqrt{\frac{\log(d)}{d}}$, $\gamma_2=\frac{5}{d^{1.5}}$, 
$N = k^{1/3}/32$, $L = \frac{d}{32C_2 k\log(d)}$ and $\rho = \frac{1}{ \sqrt{NL}}$.
For some set of vectors $x_{1:N,1:L} \in (\re^d)^{N\times L}$, we define the objective,
\begin{align}\label{eq:detalg-objective}
    F(w) = \max\big\{\underbrace{\max_{j\in[N],l\in [L]}\bc{\ip{w}{x_{j,l}} - j\gamma_1 - l\gamma_2}}_{\text{Nemirovski: } \nem(w)},~ \underbrace{4d^{6}\|A w\|_\infty - \rho}_{\text{barrier: } \barr(w)} \big\}.
\end{align}

\begin{algorithm}[h]
\caption{Optimization process with adversarially chosen objective}
\label{alg:detalg-opt}
\begin{algorithmic}[1]
\STATE $\forall j\in[N],l\in[L]: ~\cK_{j,l} \sim \Unif\big(\Gr(3kl, d)\big)$, and  $\cV_{j,l}$  a minimal $(d^{-1.5})$-cover of $\cK_{j,l}\cap \cB(1)$ 
\STATE Sample $A \in \re^{d'\times d}$ as matrix of $d'=d/2$ uniformly random unit vectors 
\STATE Round length $T = \frac{d}{C_1 \log(d)}$ ~~\textit{(recall $C_1$ is as in Proposition \ref{prop:msgh-bound})} \label{line:start-of-loops}
\FOR{$j\in[N-1]$} 
\STATE Set $\cL_{j,0}$ as the empty subspace
\FOR{$l \in [L]$}
\STATE \textit{We refer to this inner loop as round $(j,l)$}
\STATE $\cL_{j,l} = \Span(\cL_{j,l-1}, \cL_{j,l}')$ where $\cL_{j,l}' \in \Gr([2k],d)$ is chosen by some process (see Lemma \ref{lem:detalg-opt-ocvg})
\STATE $\hat{x}_{j,l} \sim \Unif(\cK_{j,l} \cap \cL_{j,l}^\perp \cap \cS(1))$ and $x_{j,l} = \argmin_{v\in\cV_{j,l}}\{\|v-\hat{x}_{j,l}\|\}$
\FOR{$t \in [T]$}
\STATE $\opt$ sends query $\ujlt$ and receives $\tilde{\cO}_{j,l}(\ujlt)$, where $\tilde{\cO}_{i,j}$ is as defined via Eqn. \eqref{eq:detalg-resisting-oracle}.
\ENDFOR
\ENDFOR
\ENDFOR
\STATE $\opt$ releases solution $\widehat{w}$
\STATE Sample $x_{N,1},\dots,x_{N,L} \sim \Unif(\cS(1))$
\end{algorithmic}
\end{algorithm}

Next, Algorithm \ref{alg:detalg-opt} shows how the objective function is chosen adaptively to $\opt$ by adversarially selecting the Nemirovski vectors.
Note the vectors are picked via a process that depends on a ``resisting oracle'' (defined below) rather than the true oracle. The bulk of the proof will be using the MSGH to show that this resisting oracle gives the same response the true oracle would, where the true oracle is determined by the realization of $F$ at the end of Algorithm \ref{alg:detalg-opt}.

The objective function is roughly built in the following manner. Consider some fixed $j\in[N]$. As we loop over $l\in [L]$, we start tracking the marked subspaces, which are determined by $A$ and the optimizer's memory state at the beginning of each round $(j,l)$. We sample the Nemirovski vector $x_{j,l}$ orthogonal to all previously marked subspaces.   
In order to ensure the vectors do no leak too much information in later rounds, they are ``wrapped'' in the randomly selected cover $\cV_{j,l}$. This allows the vector to have a marginal distribution that behaves like $\Unif(\cL_{j,l}^\perp \cap \cS(1))$, while at the same time allowing us to specify the vector with only $\tilde{O}(kl)$ bits that depend on $A$, rather than $\tilde{O}(d)$.

\paragraph{Resisting oracle and consistency via MSGH.}
For the resisting oracle, we take $\tilde{\cO}_{j,l}$ to be a first order oracle for the function,
\begin{align} \label{eq:detalg-resisting-oracle}
    F_{j,l}(w) = \max\big\{\max_{(j',l')\in ([j-1]\times [L]) \cup (\{j\} \times [l])}\bc{\ip{w}{x_{j',l'}} - j'\gamma_1 - l'\gamma_2},~ 4d^{6}\|A w\|_\infty - \rho\big\}. %
\end{align}
For any $w\in\re^d$, we assume the true oracle $\cO$ maps $w$ to the lowest index Nemirovski vector in $\partial \nem(w)$ or the lowest index row of $4d^6A$ in $\partial \barr(w)$ respectively; we use the analogous strategy for $\tilde{\cO}_{j,l}$.
The main ingredient in showing the resisting oracle is consistent with the true oracle is showing that each round is a play of the MSGH. 
\begin{lemma}\label{lem:detalg-opt-ocvg}
Let $j\in[N-1], l\in[L]$. 
There exists a strategy for choosing $\cL_{j,l}' \in \Gr([2k],d)$ such that,
\begin{align} \label{eq:detalg-opt-ocvg}
\underset{}{\PP}\Big[\forall t\in[T]: ~\|A\ujlt\|_\infty \geq \frac{\|\ujlt\|}{2d^{4.5}} ~\text{ or }~  \frac{\ujlt}{\|\ujlt\|} \in \cL_{j,l}' \msum\cB\Big(\frac{1}{d^{1.5}}\Big)\Big] = 1-O(1/d^3).
\end{align}
The probability is taken with respect to all randomness in Algorithm \ref{alg:detalg-opt}.
\end{lemma} 
The proof is in Appendix \ref{app:detalg-opt-ocvg}. Roughly, we show that the MSGH player can use $\opt$ to play the MSGH (with the same matrix $A$) and $\cL_{j,l}'$ can be set to the marked subspace. The past Nemirovski vectors are added into the memory budget, which is the primary limitation on how many vectors we can embed in the loss. The current vector, $x_{j,l}$, is used as the hint.
Using Lemma \ref{lem:detalg-opt-ocvg}, we can show the resisting oracle matches the true oracle.
\begin{lemma} \label{lem:detalg-resisting-consistent}
For Algorithm \ref{alg:detalg-opt}, $\pr{}{\forall (j,l,t)\in[N]\times[L]\times[T]: \tilde{\cO}_{j,l}(\ujt) = \cO(\ujt)} \geq 1-O(\frac{1}{d})$. %
\end{lemma}
The proof is deferred to Appendix \ref{app:detalg-resisting-consistent}. In the proof, there are two cases of interest. For query $\ujlt$, if $\ujlt \notin \cL'_{j,l} \msum \cB(d^{-1.5})$, then the barrier function dominates, which is the same between $F_{j,l}$ and $F$. Alternatively, if $\ujlt \in \cL'_{j,l} \msum \cB(d^{-1.5})$, then for any $i \geq l$, $|\ip{\ujlt}{x_{j,i}}| \leq 2d^{-1.5}$ because each $x_{j,i}$ is sampled orthogonal to $\cL'_{j,l}$.
The scale of $\gamma_2$ then prevents such vectors from being in the subgradient of $F$. Concentration inequalities will also show that, for $j'>j$ and $l'\in [L]$, $x_{j',l'}$ also cannot be in the subgradient due to the scale of $\gamma_1$ and the fact that the vectors are sampled over a sufficiently high dimensional space.

\paragraph{Finishing the proof of Theorem \ref{thm:detalg-lb}.}
A final lemma we will need is that the constructed objective function has a stronger minimizer for any choice of $\opt$.
\begin{lemma}\label{lem:detalg-minimizer}
With probability at least $1-2^{-\Omega(d)}$ there exists $w^*\in\cB(1)$ such that
$F(w^*) \leq -\frac{\sqrt{\log(d)}k^{1/3}}{\sqrt{d}}$.
\end{lemma}
We defer the proof to Appendix \ref{app:detalg-minimizer}, and here finish the proof of the lower bound. %

\begin{proof}[Proof of Theorem \ref{thm:detalg-lb}]
To finish the proof of Theorem \ref{thm:detalg-lb}, recall $n \leq \frac{d^{8/3}}{C_3\log(d)m^{2/3}}$ which is at most $T(N-1)L$ for $C_3$ large enough.
Thus running Algorithm \ref{alg:detalg-opt} with $\opt$ implies that the candidate solution $\widehat{w}$ is independent of $x_{N,1}$, since the resisting oracle responses do not depend on $x_{N,1}$. A standard concentration inequality (Lemma \ref{lem:random-correlation}) then implies that with probability at least $1-O(1/d)$ the output of $\opt$ satisfies $F(\widehat{w}) \geq -N\gamma_1 - L\gamma_2 - 2\sqrt{\frac{\log(d)}{d}} \geq -15 N\sqrt{\frac{\log(d)}{d}} \geq -\frac{\sqrt{\log(d)}k^{1/3}}{3\sqrt{d}}$. Now by Lemma \ref{lem:detalg-minimizer} above, with probability at least $1-O(1/d)-2^{-\Omega(d)}$ we have $F(\widehat{w}) > F(w^*) + \alpha$.  
Further, by Lemma \ref{lem:detalg-resisting-consistent}, w.p. at least $1-O(1/d)$ the oracle responses given by the resisting oracle match the true oracle, and under this event the candidate solution $\widehat{w}$ is the same as it would be had $\opt$ run using the true oracle. 
By a union bound, the previously described events all happen with non-zero probability provided $d$ is larger than some constant. Consequently, there exists an objective function where $\opt$ does not output an $\alpha$-suboptimal solution.
Finally, the above holds for any $m\geq d^{8/5}$, and thus evaluating the bound at $m=d^{8/5}$ yields a lower bound for the low memory regime.
\end{proof}

\section{Lower Bound for Randomized Algorithms} \label{sec:randomized}
We now present our lower bound for randomized algorithms. 
\vspace{-5pt}\begin{theorem} \label{thm:main-lb}
Let $C_4$ be a universal constant, and assume $d$ is larger than some constant.
Let
$m \in [d\log(d), \frac{d^2}{256\log(d)}]$,
$n \leq \frac{d^2}{C_4\sqrt{m}\log(d)}$
and $\opt \in \cM_{\mathsf{rand}}(m,n)$.
Then there exists a $1$-Lipschitz convex loss function such that $\opt$ does not output a solution, $\widehat{w}\in\cB(1)$, satisfying
$F(\widehat{w})-\min_{w\in\cB(1)}\{F(w)\} \leq \frac{1}{d^{7}}$
with probability at least  $1-O(\frac{1}{d})$.
\end{theorem}
\vspace{-5pt}The proof is provided in the next section, but we provide some discussion here first. In contrast to our proof for the deterministic algorithms, we use a (non-adaptive) distribution over loss instances. This greatly simplifies the optimization process (see Algorithm \ref{alg:opt}). Proving the resisting oracle is consistent in this case will use the orthogonal correlated vector game described in Proposition \ref{prop:ocvg-bound}, which will show that $\tilde{\Omega}(d)$ queries must be made before the next Nemirovski vector can be discovered. This results from the fact that OCVG will show it is hard to find a vector substantially correlated with the most recently discovered Nemirovski vector that also avoids the barrier function, and the scale of the Nemirovksi offset will be large enough to ensure that such large correlation is necessary to make progress.

In this construction, we also add a ``wall function'' to the objective (see Eqn. \eqref{eq:objective}), reminiscent of the kind used by \cite{bubeck19_highlyparallel}. Roughly, this wall function forces the component of the query outside the span of the previously discovered Nemirovski vectors to be small. This allows the Nemirovski function offset to be smaller and thus for more vectors to be embedded. However, the wall function substantially complicates the application of the OCVG, as the ``wall'' is defined using the projections of the Nemirovski vectors orthogonal to $A$. As such, its gradients may leak information about $A$. 
To control this, we argue explicitly about the conditional distribution of $A$ given the previously discovered Nemirovski vectors, and show it mostly still behaves like a collection of uniformly random vectors, but of smaller norm. 
When we argue that an OCVG player can use $\opt$ to play the OCVG (similar to Lemma \ref{lem:detalg-opt-ocvg}) we give the optimization problem a matrix, $A$, of uniformly random unit vectors, but give the OCVG instance a matrix, $A'$, of random vectors with smaller norms $\Delta_1,...,\Delta_{d'} \in [0,1]$. 
We can then write $A$ as $A'+E$, for some matrix $E\in\re^{d'\times d}$ which is known to the OCVG player and which depends on the previously discovered Nemirovski vectors.
The OCVG player then uses the game query offsets, $\eta_{1:T}$, to translate oracle queries made by $\opt$ into game queries that return the same row of $A'$ that the optimization query would return from $A$. This allows the player to correctly implement the resisting oracle. 
Ultimately, we are still able to show that making optimization progress would lead the OCVG player to win the game.  
We defer further discussion to the proof.

\subsection{Proof of Theorem \ref{thm:main-lb}}
\label{sec:proof-of-main-lb}
Our proof will use a distribution over hard loss instances. As such, throughout we will assume the optimization algorithm is deterministic, which yields the same lower bound by Yao's minimax principle. Further, by a standard reduction, it suffices to show a lower bound for algorithms which achieve $\alpha:=\frac{1}{d}$ suboptimality on $d^{6}$-Lipschitz loss functions.
\vspace{-5pt}\paragraph{Hard Instance.}
Let $k= \frac{32m}{d}$; note $k \geq 32\log(d)$ by the assumption on $m$.
Define $\gamma = 10\alpha\sqrt{\frac{k}{d}}$ and $N=\frac{1}{48}\sqrt{\frac{d}{k}}$ and $d'=d/2$.
Sample the matrix $A \in \re^{d'\times d}$ such that each row, $a_i \sim \Unif(\cS(1))$, $i\in[d']$. 
We sample $x_1,...,x_N$ in the following manner.
Let $x_1\sim\Unif(\cS(1))$ and define $\tilde{x}_1 = \proj_A^\perp x_1$. Now for each $j\in[N]$, sample 
$x_j \sim \Unif(\Span^\perp(\tilde{x}_{<j}) \cap \cS(1))$ 
and let
$\tilde{x}_j =\proj_{A}^\perp x_j$. 
Note that $\tilde{x}_1,\dots,\tilde{x}_N$ is an orthogonal set of vectors, as $\proj_A^\perp x_j = \proj_{A,\tilde{x}_{<j}}^\perp x_j$, since both the rows of $A$ and $x_j$ lie in $\Span^\perp(\tilde{x}_{<j})$.
The hard instance for the loss takes the following form:
\begin{align}\label{eq:objective}
    F(w) = \max\big\{\underbrace{\max_{j\in[N]}\bc{|\ip{w}{x_j} - 2\alpha| - j\gamma}}_{\text{Nemirovski: }\nem(w)},~ \underbrace{d^{6}\|A w\|_\infty}_{\text{barrier: }\barr(w)},~ \underbrace{\max_{z\in\cZ}\bc{\ip{w}{z}} - 4\alpha}_{\text{wall: } \wall(w)} \big\},
\end{align}
where $\cZ$ will be defined momentarily.
We will refer to the three components in the maximum as the Nemirovski function, barrier function (of \cite{MSSV22}) and wall function (a modified version of that found in \cite{bubeck19_highlyparallel}) respectively.
For the wall function, 
we define,
\vspace{-5pt}\begin{align*}
    \mathcal{C}_j = \Big\{y: \abs{\ip{\tilde{x}_j}{y}} > 6\sqrt{\frac{k}{d}}\Big\}, \quad \forall j\in[N] && \text{ and } && \cZ = \big\{z: \|z\|=1 ~~\land~~ z \notin \bigcup\limits_{j\in[N]} \cC_{j}\big\}.
\end{align*}
That is, $\cZ$ is the set outside the correlation caps, $\cC_1,...,\cC_N$. There are two main functional differences between our wall function and the one defined in \cite{bubeck19_highlyparallel}. First, a vector $w$ must have correlation $|\ip{w}{\tilde{x}_j}| \geq 6\sqrt{k/d}$ with at least one $\tilde{x}_j$ to avoid $\wall(w) = \|w\| - 4\alpha$, whereas \cite{bubeck19_highlyparallel} uses a correlation threshold of roughly $\sqrt{\log(d)/d}$. Further, the set $\cZ$ uses different vectors than those in the Nemirovski function, which leads to a more delicate analysis when controlling the information leaked by the oracle responses.

It can also be shown the vector $\alpha \sum_{j=1}^{N}s_j \tilde{x}_j$, for some $s\in[1,8]^N$, serves as a good minimizer. The proof of the following lemma can be found in Appendix \ref{app:minimizer}.
\begin{lemma} \label{lem:minimizer}
With probability at least $1-2^{-\Omega(d)}$ there exists $w^*\in\cB(1)$ such that $F(w^*) = 0$.   
\end{lemma}

\subsubsection{Applying the OCVG to the Optimization Problem}
In this section, we show how to use the OCVG (Proposition \ref{prop:ocvg-bound}) to prove a lower bound for the optimization problem given by Eqn. \eqref{eq:objective}. To do this, we show that the process defined by Algorithm \ref{alg:opt} cannot find an $\alpha$-suboptimal solution. Up to the replacement of the true gradient oracle with a ``resisting'' oracle, $\tilde{\cO}_j$ (defined below), the process in Algorithm \ref{alg:opt} is the same as the actual optimization process for an optimizer which makes $(N-1)\frac{d}{C_1 \log(d)}$ queries to the first order oracle. 
The main piece in finishing the proof is showing that with high probability the resisting oracle in fact returns the same gradient responses that the true oracle would return. To do this, we first show that the queries to the resisting oracle satisfy an OCVG like property.

\begin{algorithm}[h]
\caption{Optimization Process}
\label{alg:opt}
\begin{algorithmic}[1]
\STATE $T = \frac{d}{C_1 \log(d)}$,  where $C_1$ is as in Proposition \ref{prop:msgh-bound}
\FOR{$j\in[N-1]$}
\FOR{$t \in [T]$}
\STATE $\opt$ sends query $u_{j,t}$ to oracle and receives $\tilde{\cO}_j(\ujt)$,  where $\tilde{\cO}_j$ as defined via  Eqn. \eqref{eqn:resisting-oracle}
\ENDFOR
\ENDFOR
\STATE $\opt$ releases solution $\hat{w}$
\end{algorithmic}
\end{algorithm}

\paragraph{Resisting oracle.} For $j\in[N]$, we let $\tilde{\cO}_j$ be a first order oracle for the function,
\begin{align} \label{eqn:resisting-oracle}
    F_j(w)\! =\! \max\big\{\max_{i\in[j]}\bc{|\ip{w}{x_i} - 2\alpha| - i\gamma},~ d^{6}\|A w\|_\infty,~ \max_{z\in\cZ_j}\bc{\ip{w}{z}}\! -\! 4\alpha \big\}
\end{align}
where $\cZ_j = \big\{z: \|z\|=1 ~~\land~~ z \notin \bigcup_{i\in[j-1]} \cC_{i}\big\}$.
We defer the technical details on exactly how $\cO$ and $\tilde{\cO}_j$ choose an element from the subgradient of $F$ and $F_j$ (respectively) to Appendix \ref{app:resisting-consistent}.
In short, $F_j$ is just $F$ with the pieces that use $x_{j+1:N}$ and $\tilde{x}_{j:N}$ having been ``deleted''.
For notation, we also define $\nem_j(w) = \max_{i\in[j]}\bc{|\ip{w}{x_i} - 2\alpha| - i\gamma}$.%

\paragraph{The inner loop is a play of the OCVG in $\Omega(d)$ dimensions.}
A key fact we will use to show the resisting oracle's consistency is that each inner loop is a play of the OCVG with embedding dimension $d-j+1$. Recall we have defined the OCVG as the modification of Algorithm \ref{alg:MSGH} where $\cP$ receives a vector $x\sim \cS(1)$ in Line \ref{line:player-info} and wins via the conditions specified by Proposition \ref{prop:ocvg-bound}.
In the following, 
for each $j\in[N-1],t\in[T]$, define $\tujt = \proj_{\tilde{x}_{<j}}^\perp \ujt$. %
Further, recall that we have assumed $\opt$ is deterministic by Yao's minimax principle, so the randomness in Algorithm \ref{alg:opt} comes only from sampling the objective.
\begin{lemma}\label{lem:opt-ocvg}
Consider Algorithm \ref{alg:opt}. Fix $j\in[N-1]$. 
It holds that,
\begin{align} \label{eq:opt-ocvg}
\underset{A,x_{\leq j}}{\PP}\Big[\exists t\in[T]: \! ~\|A\tujt\|_\infty \leq \frac{\|\tujt\|}{8d^{5.5}},~  \max\{|\ip{\tujt}{x_j}|,|\ip{\tujt}{\tilde{x}_j}|\} \geq 6\|\tujt\|\sqrt{\frac{k}{d}}\Big] = O\big(\frac{1}{d^3}\big).
\end{align}
\textit{(Note also $\|A\ujt\|_\infty = \|A\tujt\|_\infty$ and $\ip{\ujt}{x_j} = \ip{\tujt}{x_j}$.)} 
\end{lemma} 
The proof is deferred to Appendix \ref{app:opt-ocvg}. The reduction is more involved than in Lemma \ref{lem:detalg-opt-ocvg}. Once again, the vectors $\tilde{x}_j$ may leak information about the loss. We control this via a careful analysis, which shows that even after receiving some number of projected Nemirovski vectors, the conditional distribution of $A$ is still roughly a collection of uniformly random unit vectors, but over a lower dimensional space. We will also leverage a lower bound on the non-zero singular values of $A$, which will let us show that if $\|A\tujt\|_\infty$ is small, then $\ip{\tujt}{x_j}$ and $\ip{\tujt}{\tilde{x}_j}$ are close. 

\paragraph{The resisting oracle is consistent.} The next step is ensure that the resisting oracle matches the true oracle responses. 
\begin{lemma}\label{lem:resisting-consistent}
For Algorithm \ref{alg:opt}, 
$\pr{A,x_{1:N}}{\forall j\in[N-1], t\in[T]: \tilde{\cO}_j(\ujt) = \cO(\ujt)} \geq 1-O(\frac{1}{d})$.
\end{lemma}
The proof works by considering the behavior of the oracle response in several cases for each query $\ujt$, $j\in[N],t\in[T]$. We defer the full proof to Appendix \ref{app:resisting-consistent} and summarize these cases here. 
\textbf{Case~1:} $\|\tujt\| \in [0,\alpha\sqrt{\frac{k}{d}}] ~~\&~~ \|A\tujt\|_\infty \geq \frac{\|\tujt\|}{8d^{5.5}}$: Since $\|\tujt\|$ is small, the scale of $\gamma$ in $\nem$ ensures it can be evaluated without $x_{j+1:N}$. Since $\|A\tujt\|_\infty$ is large, we can also show the wall function is not needed (see Lemma \ref{lem:barrier-beats-wall}).  
\textbf{Case~2:} $\|\tujt\| \in [\alpha\sqrt{\frac{k}{d}},1] ~~\&~~ \|A\tujt\|_\infty \geq \frac{\|\tujt\|}{8d^{5.5}}$: Since $\|A\tujt\|_\infty \geq \frac{\|\tujt\|}{8d^{5.5}}$, the wall function is again not needed. Further, because $\|\tujt\|$ is not too small, the barrier function will be larger than the Nemirovski function.
\textbf{Case~3:} $\|\tujt\| \in [0,3\alpha] ~~\&~~  \|A\tujt\|_\infty \leq \frac{\|\tujt\|}{8d^{5.5}}$: Since $\|A\tujt\|_\infty \leq \frac{\|\tujt\|}{8d^{5.5}}$, Lemma \ref{lem:opt-ocvg} will guarantee that $|\ip{\tujt}{x_j}| \leq 6\|\tujt\|\sqrt{\frac{k}{d}} \leq 6\alpha\sqrt{\frac{k}{d}}$. 
A standard concentration inequality also ensures
$|\ip{\tujt}{x_i}| \leq 3\|\tujt\|\sqrt{\frac{\log(d)}{d}} \leq 3\alpha\sqrt{\frac{k}{d}}$ for any $i > j$.
The scale of $\gamma$ then guarantees $\nem$ can be evaluated without $x_{j+1:N}$. Similar to the preceding, we can also obtain $|\ip{\tujt}{\tilde{x}_i}| \leq  6\|\tujt\|\sqrt{\frac{k}{d}}$ for all $i\geq j$, which can be used to show that $\tilde{x}_{j:N}$ are not need to evaluate $\wall(\ujt)$. This fact is not immediate, but uses a similar analysis to \cite{bubeck19_highlyparallel} (see Lemma \ref{lem:wall}). 
\textbf{Case 4:} $\|\tujt\| \in [3\alpha,1] ~~\&~~ \|A\tujt\|_\infty \leq \frac{\|\tujt\|}{8d^{5.5}}$: This is the case which crucially benefits from the wall function. Because Lemma \ref{lem:opt-ocvg} ensures $\max\{|\ip{\tujt}{\tilde{x}_j}|\} \leq 6\|\tujt\|\sqrt{\frac{k}{d}}$, but $\|\tujt\|$ is not too small, the wall function ensures $\max\{\wall(\ujt),\nem_j(\ujt)\} \geq \nem(\ujt)$. As in Case 3, the wall function can be evaluated without $\tilde{x}_{j:N}$.

\paragraph{Concluding the proof of Theorem \ref{thm:main-lb}.}
Recall that by Yao's minimax principle, it suffices to prove a lower bound for a deterministic algorithm under a distribution over hard loss instances.
Thus let $\opt \in \cM_{\mathsf{det}}(m,n')$ for $n' \leq (N-1)\frac{d}{C_1\log(d)}$.
Lemma \ref{lem:resisting-consistent} shows that for any such $\opt$, with probability at least $1-O(1/d)$ (under the draw of the objective function) that the resisting oracle is consistent with the true gradient oracle responses when the loss is drawn from the distribution specified at the start of Section \ref{sec:proof-of-main-lb}.
This implies that with probability at least $1-O(1/d)$ $\opt$ will output the same solution as it would if run with the true oracle. 
Since we only wish to show that $\opt$ fails to output an $\alpha$-suboptimal solution using the true oracle w.p. at least $1-O(1/d)$, what remains is to show that the solution released by $\opt$ after Algorithm \ref{alg:opt} terminates is not $\alpha$-suboptimal w.p. at least $1-O(1/d)$.
To see this, observe that since $\hat{w}$ is independent of $x_N$, then a standard concentration inequality (Lemma \ref{lem:random-correlation}) implies that with probability at least $1-O(1/d)$ one has $|\max\{\ip{\hat w}{x_N},\ip{\hat w}{\tilde{x}_N}\}| \leq 2\|\proj^\perp_{\tilde{x}_{<N}} \hat w\|\sqrt{\frac{\log(d)}{d}}$. 
If $\|\proj^\perp_{\tilde{x}_{<N}} \hat w\| \leq 6\alpha$, then $\nem(\hat w) \geq 2\alpha - 12\alpha\sqrt{\frac{\log(d)}{d}} - N\gamma > \alpha$, provided $d$ is larger than some constant. 
Alternatively if $\|\proj^\perp_{\tilde{x}_{<N}} \hat w\| \geq 6\alpha$, then $\wall(\hat w) \geq 2\alpha$ by Lemma \ref{lem:wall} (Appendix \ref{app:resist-consist-lemmas}). 
By Lemma \ref{lem:minimizer}, the minimizer, $w^*$, satisfies $F(w^*) \leq 0$ w.p. at least $1-2^{\Omega(d)}$.
This implies that $\hat w$ is not an $\alpha$-suboptimal solution under an event that happens with probability at least $1-O(1/d)$.
Finally, we observe the number of oracle calls is
$T(N-1) = \Omega\Big(\frac{d}{\log(d)} \cdot \sqrt{\frac{d}{k}}~\Big) = \Omega\br{\frac{d^2}{\sqrt{m}\log(d)}}.$ Thus for some constant $C_4=O(1)$, $n = \frac{d}{C_4 \sqrt{m}\log(d)} \leq (N-1)\frac{d}{C_1\log(d)}$ as desired.

\appendix
\section{Supplementary Lemmas} \label{app:lemmas}
We will use the following lemma about random projection frequently. A standard reference is \cite[Lemma 5.3.2]{vershynin_hdp}, although constant factors can be obtained from \cite[
Chapter 7]{concentration-inequalities}.
\begin{lemma}(Random Projection) \label{lem:random-projection}
Let $k\in[d]$, $V\in \Gr(k,d)$ and $x\sim\Unif(\cS(1))$. Then 
\begin{align*}
    \pr{}{(1-\epsilon)\sqrt{\frac{k}{d}} \leq \|\proj_V x\| \leq (1+\epsilon)\sqrt{\frac{k}{d}} } \geq 1-2e^{-k\epsilon^2/2}
\end{align*}
\end{lemma}
\noindent Note the same bound holds for fixed $x$ with $V\sim \Unif(\Gr(k,d))$. The following concentration inequality follows as a corollary for the case when $k=1$.

\begin{lemma}(Random Correlation) \label{lem:random-correlation}
Let $u\in\cB(1)$ and $c>0$ and $x\sim\Unif(\cS(1))$. Then,
\begin{align*}
    \pr{}{|\ipnos{u}{x}| \geq \sqrt{\frac{c\log(d)}{d}}} \leq \frac{2}{(d)^{c/2}}.
\end{align*}
\end{lemma}
Note this implies a concentration inequality when $x$ is a uniformly random unit vector over a linear subspace of $\re^d$ as well, as the inner product is then determined by the component of $u$ in this subspace.

We will also use the following bounds on the size of packings/covers.
\begin{lemma}(\cite[Chapter 4.2.1]{vershynin_hdp})\label{lem:packing-ball}
Let $\epsilon \in [0,1]$, $\cG\subseteq \cB(1)$ and let $\cV$ be either an $\epsilon$-packing or a minimal $\epsilon$-cover of $\cG$. Then $|\cG| \leq (3/\epsilon)^n$.     
\end{lemma}

The following lemma quantifies the extent to which a set of vectors behaves like a basis, in some sense.
\begin{lemma}(\cite{RV09}, \cite[Lemma 1.11]{sing-vals-of-random-mats}) \label{lem:singular-vals}
If the rows of $U\in \re^{d' \times d}$, $\{u_1,\dots,u_{d'}\}$ satisfy 
$\|\proj_{u_{\neq j}} u_j\| \geq \tau$ for every $j\in [d']$, the minimum non-zero singular value of $U$ satisfies $\sigma_{min} \geq \frac{\tau}{\sqrt{d}}$.
\end{lemma}

This is particularly helpful in conjunction with the following lemma.
\begin{lemma}\label{lem:l1-coeff-norm}
Let $d' \le d$, and let $u_1,...,u_{d'} \in \re^d$ be linearly independent vectors and $U \in \re^{d \times d'}$ be the corresponding matrix. Assume the minimum singular value of $U$ is $\sigma_{min}$. Then for any $v \in \Span(u_{1:d'})$, the unique $c \in \re^{d'}$ such that $v = Uc$ satisfies $\|c\|_1 \leq \frac{\sqrt{d} \|v\|}{\sigma_{min}}$.
\end{lemma}
\begin{proof}
Since $v \in \Span(U)$, %
there exists a unique ${c}\in\re^{d'}$ such that ${v} = {U}{c}$. We then have $\|v\| \ge \sigma_{min} \|c\|$ and $\|c\|_1 \le \sqrt{d'}\|c\|$, from which the lemma follows.\end{proof}

\section{Supplement to Section \ref{sec:game}} \label{app:game}

\subsection{Proof of Lemma \ref{lem:no-large-ortho-space}} \label{app:no-large-ortho-space}
Let $\cA = \{A' \in \Supp(A): \Dim(g(A')) \geq k\}$. Now consider any fixed $A' \in \cA$. Letting $a_j$ denote the $j$'th row of $A$, we have,
\begin{align*}
    \pr{A}{\max_{u\in g(A') \cap \cS(1)} \|A u\|_\infty \leq \sqrt{\frac{k}{4d}}
    } 
    &\leq \pr{A}{\forall j\in[d']: \max_{u \in g(A')\cap \cS(1)}\{|\ip{u}{a_j}|\} \leq \sqrt{\frac{k}{16d}}} \\
    &\leq \Pi_{j=1}^{d'} \pr{x\sim\Unif(\cS(0.5))}{\max_{u \in g(A')\cap \cS(1)}\{|\ip{u}{x}|\} \leq \sqrt{\frac{k}{16d}}} \\
    &\leq 2e^{-d'k/8}.
\end{align*}
The last inequality follows from the concentration results for random projection (see e.g. \cite[Lemma 5.3.2]{vershynin_hdp}) and the fact that each $a_j$ has $\|a_j\| \geq 0.5$.
Since $|\Range(g)| \leq 2^{kd'/16}$, 
we have $|\{g(A'): A'\in\cA\}| \leq 2^{kd'/16}$.
Taking a union bound over all $A' \in \Supp(A)$ and recalling $\log(|\Range(g)|) \leq kd'/16$,
\begin{align*}
    &\pr{A}{\exists A'\in\cA: \max_{u\in g(A')\cap \cS(1)} \|Au\|_\infty \leq \sqrt{\frac{k}{4d}}
    } \leq e^{\frac{kd'}{16}-\frac{d'k}{8}} \leq e^{-d'k/16}. 
\end{align*}

Since the event $\bc{\Dim(g(A)) \geq k \land \max\limits_{u\in g(A)\cap \cS(1)}\{\|A u\|_\infty\} \leq \sqrt{\frac{k}{4d}}}$ can only happen if both $A \in \cA$ and $\exists A'\in\cA: \max\limits_{u\in g(A')\cap \cS(1)}\{\|A u\|_\infty\} \leq \sqrt{\frac{k}{4d}}$, the lemma is proven.

\subsection{Proof of Lemma \ref{lem:slice-dim-yields-ortho-space}}
\label{app:slice-dim-yields-ortho-space}

\begin{lemma*}
Let $\gamma > 0$, $\hat{d}\in[d]$ and $\cG\subseteq \re^d$.
Then there exists a function $f_{\cG}:\re^{d'\times d} \mapsto \Gr(\hat{d},d)$ such that $\log(|\Range(f_{\cG})|) \leq \hat{d}\log(|\cG|)$ and
whenever there exists $\cY\subseteq \cG$ such that 
$\forall v\in\cY: \|Av\|_\infty \leq \gamma$, and
$\Gamma(\cY,\tau)\geq\hat{d}$, then
$\max_{v\in f_\cG(A)}\{\|Av\|_\infty\} \leq \|v\|\frac{\gamma d}{\tau}.$
\end{lemma*}
\begin{proof}
    Let $u_1, \ldots, u_{\hat{d}} \in \cY$ maximize the $\hat{d}$-dimensional volume of the parallelepiped $P(u_1, \ldots, u_{\hat{d}}) = \left\{\sum_{j=1}^{\hat{d}} c_j u_j: c_1, \ldots, c_{\hat{d}} \in [0,1]\right\}$ over all choices of $\hat{d}$ vectors from $\cY$. By the ``base times height'' formula for volume, we have that, for each $j \in [\hat{d}]$,
    \[
    \mathsf{Vol}_{\hat{d}}(P(u_1, \ldots, u_{\hat{d}}))
    = 
    \|\proj_{u_{\neq j}}^\perp u_j\|\, \mathsf{Vol}_{\hat{d}-1} (P(u_{\neq j})),
    \]
    where $P(u_{\neq j})$ is the parallelepiped spanned by all the chosen vectors except $u_j$. Since the volume of $P(u_1, \ldots, u_{\hat{d}})$ is maximal, it cannot be increased by replacing $u_j$ with any other vector in $\cY$, so we have,
    \begin{equation}
        \|\proj_{u_{\neq j}}^\perp u_j\| \ge \|\proj_{u_{\neq j}}^\perp v\| \quad \forall v \in \cY.
    \end{equation}
    Therefore, $\cY$ is contained in $\Span(u_{\neq j}) \msum \cB(\|\proj_{u_{\neq j}}^\perp u_j\|)$. Since $\hat{d} = \Gamma(\cY,\tau)$, and $\Span(u_{\neq j})$ is $\hat{d}-1$ dimensional, this implies that $\|\proj_{u_{\neq j}}^\perp u_j\| > \tau.$

    We can now apply Lemmas~\ref{lem:singular-vals}~and~\ref{lem:l1-coeff-norm} to the matrix $U$ whose columns are $u_1, \ldots, u_{\hat{d}}$, and get that any vector $v$ in $\Span(u_{1:\hat{d}})$ can be written as $v = \sum_{j=1}^{\hat{d}} c_j u_j$ with $\|c\|_1 \le \frac{d\|v\|}{\tau}$. Therefore, for any $v \in \Span(u_{1:\hat{d}})$ we can write,
    \[
    \|Av\|_{\infty} \le \sum_{j = 1}^{\hat{d}} |c_j| \|Au_j\|_{\infty} \le \frac{d\|v\|\gamma}{\tau}.
    \]
    To complete the proof, we choose $f_{\cG}(A) = \Span(u_{1:\hat{d}})$ whenever the subset $\cY$ exists. Since there are at most ${\cG \choose \hat{d}}\le |\cG|^{\hat{d}}$ choices of $\{u_1, \ldots, u_{\hat{d}}\}$, the requirement on the size of the range of $f_{\cG}$ is satisfied. When the subset $\cY$ does not exist, we take $f_{\cG}$ to map to the span of the first $\hat{d}$ elements of $\cG$ (the choice of this subspace does not matter).
\end{proof}

\subsection{Proof of Lemma \ref{lem:ortho-set-small-width}} \label{app:ortho-set-small-width}
For convenience, we recall the definition of $\cG_b$,
\begin{align*}
 \cG_b = \Big\{u\in\cS(1): \bigpr{A|B=b}{\|\tilde{A} u\|_\infty \leq \frac{1}{d^{4.5}}} \geq 2^{-d/c_1} \Big\},   \quad \forall b\in\bit^{m_1}.
\end{align*}

Before proving the lemma, we introduce some auxiliary lemma and additional notation.
For any $b\in \bit^m$, we define packing-cover, 
$\cG'_b \subseteq \cG_b$ which is a  maximal $\epsilon'$-packing (and $2\epsilon'$-cover) of $\cG_b$, with $\epsilon' = \frac{1}{d^{3}}$. The set $\cG_b'$ can be arbitrarily chosen, but is assumed to be a deterministic function of $\cG_b$. 
Before proving Lemma \ref{lem:ortho-set-small-width}, we will show $\cG_B'$ is usually not too big. 
The following lemma will be useful.
\begin{lemma} \label{lem:packnum-lb-linear-dim}
Let $\cY$ be a $\tau$-packing of $\cS(1)$ of size at least $2^{r}$, $r>0$.
Then $\Gamma(\cY, \tau/2) \geq \frac{r}{\log(6/\tau)}$. 
\end{lemma}
\begin{proof}
Let $\cL$ be some linear subspace.
We start by establishing that if points $v,v' \in \cY$ are within distance $\tau/2$ to $\cL$, then
$\|\proj_{\cL} v - \proj_{\cL} v'\| \geq \tau \sqrt{3/4}$, so the projection of $\cY$ onto $\cL$ is a $(\tau\sqrt{3/4})$-packing of $\proj_{\cL} \cS(1)$.
To see this, observe
\begin{align*}
    \tau^2 \leq &\|v-v'\|^2 \leq \|\proj_{\cL} v - \proj_{\cL} v'\|^2 + \|\proj_{\cL}^\perp v - \proj_{\cL}^\perp v'\|^2 \leq \|\proj_{\cL} v - \proj_{\cL} v'\|^2 + \frac{\tau^2}{4} \\
    &\implies \|\proj_{\cL} v - \proj_{\cL} v'\|^2 \geq \frac{3\tau^2}{4}
\end{align*}

Now observe that, for any linear subspace, $\cL$, of dimension $d'$ we have that $\proj_{\cL}\cS(1)$ is a $d'$-dimensional ball of radius $1$.
Thus its $(\tau\sqrt{3/4})$-packing number is at most $(6/\tau)^{d'}$ by Lemma~\ref{lem:packing-ball}.
This implies
$(6/\tau)^{d'} \geq 2^{r}$ and thus $d' \geq \frac{r}{\log(6/\tau)}$ as claimed. 
\end{proof}

We now can prove that the packing-cover is not too big with high probability.
\begin{lemma} \label{lem:prob-in-bgood}
It holds that 
$\pr{A}{|\cG_B'| \leq 2^{d/100}} \geq 1-2^{-\Omega(d^2/\log(d))}$. 
\end{lemma}
\label{app:prob-in-bgood}
\begin{proof}
We will show for some $\hat{d}\in[d]$ that when $\cG_B'$ is large, it is possible to map $A$ to a $\hat{d}$ dimensional subspace nearly orthogonal to $\tilde{A}$ with much less than $d' \hat{d}$ bits, which by Lemma \ref{lem:no-large-ortho-space} is a contradiction unless $\pr{B}{|\cG_B'| \geq 2^{d/100}}$ is small. Recall $B$ is a deterministic function of $A$, and also that $\tilde{A}$ is the matrix of unobserved rows of $A$.

In the following, consider the event $B=b$ for any $b$ such that $|\cG_b'| \geq 2^{d/100}$. 
Let $\cG_b'' \subseteq \cG_b'$ be a deterministically chosen subset of $\cG_b'$ of size $2^{d/100}$. 
Let 
$\cY_{b,\tilde{A}}'' =|\{v\in\cG_b'': \|\tilde{A} v\|_\infty \leq \frac{1}{d^{4.5}} \}|$.  
Since $v\in\cG_b'' \subseteq \cG_b' \subseteq \cG_b$, by the definition of $\cG_b$ we have
$\forall v\in \cG_b'', \pr{A|b}{v\in \cY_{b,\tilde{A}}''} \geq 2^{-\frac{d}{c_1}}$.
This implies
$\ex{A|b}{|\cY_{b,\tilde{A}}''|} = \sum_{v \in \cG_b''} \pr{A|b}{v\in \cY_b''} \geq 2^{\frac{d}{100}-\frac{d}{c_1}} \geq 2^{\frac{d}{200}}$ (for $c_1$ large enough).
Since $|\cG''_b| = 2^{d/100}$ with probability $1$, we always have $|\cY''_{b,\tilde{A}}| \leq 2^{d/100}$.
We now have for any $b$ such that $|\cG'_b| \geq 2^{d/100}$, (the following expectations are taken w.r.t. $A$ given $B=b$), 
\begin{align}\label{eq:yba-bound}
   &\pr{A|B=b}{|\cY_{B,\tilde{A}}''| \geq \frac{1}{2}\ex{A|B=b}{|\cY_{B,\tilde{A}}''|}}2^{d/100} + \frac{1}{2}\ex{A|B=b}{|\cY_{B,\tilde{A}}''|} \geq \ex{A|B=b}{|\cY_{B,\tilde{A}}''|} \nonumber \\
   &\implies \pr{A|B=b}{|\cY_{B,\tilde{A}}''|\geq 2^{d/200-1}} \geq 2^{-d/200}. 
\end{align}

Above, we use the fact that for any random variable $X$,  we can construct another random variable $Y$ satisying $\E[X] \le \E[Y]$ via \(Y=\ind\{X\geq \E[X]/2\}\max\limits_{x\in\Supp(X)}\{x\} + \ind\{X < \E[X]/2\}\ex{}{X}/2.\)

Under the event $|\cG_B'| > 2^{d/100}$ and $|\cY_{B,\tilde{A}}''|\geq 2^{d/200-1}$, by Lemmas \ref{lem:packnum-lb-linear-dim} and \ref{lem:slice-dim-yields-ortho-space}, there exists a linear subspace, $\cL$, of dimension at least $\hat{d}:=\frac{d-200}{200\log(6/\epsilon')}$ such that $\max_{v\in\cL}\{\|\tilde{A}v\|_\infty\} \leq (\frac{1}{d^{4.5}})\frac{d}{\epsilon'} \leq \frac{1}{\sqrt{d}}$.
Further, $\cL$ can be identified with $\hat{d} \log_2(|\cG_{B}''|)$ bits. Since $\cG_{B}''$ is of size at most $2^{d/100}$ and itself is specified with at most $m_1$ bits, $\cL$ can be specified with $\frac{d \hat{d}}{100} + m_1$ bits. 
Consider some fixed $\sigma \subseteq [d']$, $|\sigma| \geq d'-T$,
and let $A_\sigma$ denote the restriction of $A$ to the rows indexed by $\sigma$.
Lemma \ref{lem:no-large-ortho-space} implies that no function of at most $m' :=\frac{(d'-T)\hat{d}}{16} \leq \frac{d'\hat{d}}{32}$ bits that depend on $A_\sigma$ can generate a $\hat{d}$ dimensional subspace satisfying $\max_{u\in \cL\cap\cS(1)}\{\|A_\sigma u\|_\infty\} \leq \frac{1}{d} \leq \sqrt{\frac{k}{16d}}$ with probability more than $2^{-\Omega((d'-T) \hat{d})}=2^{-\Omega(d^2/\log(d))}$. 
We indeed have $\frac{d\hat{d}}{100}+m_1 \leq \frac{d\hat{d}}{100}+\frac{d^2}{64\log(d)} \leq m'$.
By a union bound, the probability that a function of $m'$ bits can generate a subspace orthogonal to \textit{any} $A_{\sigma}$ for $\sigma \subset [d']$ and $|\sigma|\geq d'-T$, is at most $\sum\limits_{p=d'-T}^{d'}{{d'}\choose{p}} 2^{-\Omega(d^2/\log(d))} = 2^{-\Omega(d^2/\log(d))}$ provided $d$ is larger than some constant. 
Since there exists $\tilde{\sigma} \subseteq [d']$ such that $\tilde{A} = A_{\tilde{\sigma}}$, 
using Eqn. \eqref{eq:yba-bound}
it must be that
$2^{-\frac{d}{200}} \pr{}{|\cG_B'| \geq 2^{d/100}} \leq 2^{-\Omega(d^2/\log(d))}$ and so
$\pr{A}{|\cG_B'| \geq 2^{d/100}} \leq 2^{-\Omega(d^2/\log(d))}$.
\end{proof}

We now prove Lemma \ref{lem:ortho-set-small-width}.
\begin{lemma*}
For any $A'\in\Supp(A)$ and $b\in \bit^m$, let 
$\cY_{b,A'} := \{v\in\cG_b: \|A'v\|_\infty \leq \frac{1}{d^{4.5}}\}$. It holds that
$\pr{A}{\Gamma(\cY_{B,A}, \frac{1}{2d^{1.5}}) \leq k} \geq 1-2^{-\Omega(dk)}$.
\end{lemma*}
It is worth briefly emphasizing that $\cY_{B,A}$ pertains to points orthogonal to $A$ (not $\tilde{A}$), in contrast to the set $\cY''_{B,\tilde{A}}$ used in the proof of Lemma \ref{lem:prob-in-bgood}.
\begin{proof}[Proof of Lemma \ref{lem:ortho-set-small-width}]
Define $\cY_{b,A'}' = \{v \in \cG_b': \|A' v\|_\infty \leq 3\epsilon'\}$ for any $b,A'$ in their range. This set is an analog of $\cY_{b,A}$ for the packing-cover, $\cG_b'$, with a weaker orthogonality condition. We will start by analyzing the slice dimension of $\cY_{B,A}'$.

Condition on $|\cG'_B| \leq 2^{d/100}$, which happens with probability at least $1-2^{-\Omega(d^2/\log(d))
}$ by Lemma \ref{lem:prob-in-bgood}.  
Let $\hat{k}=\Gamma(\cY_{b,A}', \frac{1}{3d^{1.5}})$.
By Lemma \ref{lem:slice-dim-yields-ortho-space}, there then exists a mapping which uses an additional $\log(|\cG'_B|)\hat{k} \leq \frac{d\hat{k}}{100}$ 
bits about $A$ that produces a $\hat{k}$-dimensional linear subspace, $\cL$, satisfying 
\[\max_{v\in\cL}\{\|Av\|_\infty\} \leq 2\epsilon'd(3d^{1.5}) \leq \frac{1}{\sqrt{d}},\] (recall $\epsilon' = \frac{1}{d^{3}})$. 
Since $\cG_b$ takes at most $m_1$ bits to specify, this whole process uses at most $m_1 + \frac{d\hat{k}}{100} \leq \frac{dk}{32} + \frac{d\hat{k}}{100}$ bits. 

Now for $\hat{k} > k$, Lemma \ref{lem:no-large-ortho-space} 
implies that
generating the previously described subspace $\cL$
with $\frac{dk}{32} + \frac{d\hat{k}}{100} < \frac{d\hat{k}}{16}$ bits cannot succeed with probability better than $2^{-\Omega(d\hat{k})}$.
The process succeeds whenever $\hat{k} \geq k$ and $|\cG_B| \leq 2^{d/100}$, thus we have,
$\pr{A}{\Gamma(\cY_{b,A}', \frac{1}{3d^{1.5}}) > k,~ |\cG'_B| \leq 2^{d/100}} \leq 2^{-\Omega(d\hat{k})} \leq 2^{-\Omega(dk)}$.
Since $\pr{A}{\Gamma(\cY_{b,A}', \frac{1}{3d^{1.5}}) > k,~ |\cG'_B| > 2^{d/100}} \leq \pr{A}{|\cG'_B| > 2^{d/100}} \leq 2^{-\Omega(d^2/\log(d))}$ by Lemma \ref{lem:prob-in-bgood},
this implies $\Gamma(\cY_{b,A}', \frac{1}{3d^{1.5}}) = \hat{k} \leq k$ with probability at least $1-2^{-\Omega(d k)}$.

We now show that if $v\in \cY_{B,A}$, then there exists $v'\in\cY_{B,A}'$ such that $\|v'-v\|\leq 2\epsilon$, which will establish $\Gamma(\cY_{B,A},\frac{1}{2d^{1.5}}) \leq \Gamma(\cY_{B,A},\frac{1}{3d^{1.5}} + 2\epsilon') \leq \Gamma(\cY_{B,A}',\frac{1}{2d^{1.5}})$ and complete the proof.
Indeed, for any $v\in\cY_{B,A}$, there is some $v\in\cB(2\epsilon',v)\cap \cG_B'$ since $\cG_B'$ is a $2\epsilon'$ cover of $\cG_B$.  
The $1$-Lipschitzness of $v\mapsto \|Av\|_\infty$ then implies $\|A v'\|_\infty \leq \frac{1}{d^{4.5}} + 2\epsilon' \leq 3\epsilon'$ and thus $v'\in\cY_{B,A}'$ as well.
\end{proof}

\subsection{Proof of Lemma \ref{lem:bits-bound-posterior}} \label{app:bits-bound-posterior}
\begin{lemma*}
Let $A,B$ be random variables where $\Supp(B)=\bit^s$, $s\in\cZ^+$. Then for any $\delta \in [0,1]$,
$\pr{B}{\exists \cE \subseteq \Supp(A): \pr{}{A\in\cE | B} \leq 2^{s+\log(1/\delta)}\pr{}{A\in\cE}} \leq \delta.$
\end{lemma*}
The proof uses similar ideas to \cite[Theorem 17]{dwork2015generalization}.
\begin{proof}
Let $\mathsf{Bad} \subset \Supp(B)$ be the set such that $b \in \mathsf{Bad}$ if there exists $\cE \in \Supp(A)$ such that 
$\frac{\pr{}{B=b|A\in \cE}}{\pr{}{B=b}} \geq 2^s / \delta$. 
Since $\pr{}{B=b|A \in \cE} \leq 1$, we have that for any $b \in \mathsf{Bad}$ that $\pr{}{B=b} \leq \delta/2^s$.
Now since 
$\frac{\pr{}{B=b|A\in\cE}}{\pr{}{B=b}} = \frac{\pr{}{B=b,A\in\cE}}{\pr{}{B=b}\pr{}{A\in\cE}} = \frac{\pr{}{A\in\cE|B=b}}{\pr{}{A\in\cE}}$ 
and $|\mathsf{Bad}| \leq 2^s$, the claim follows.
\end{proof}

\subsection{Proof of Lemma \ref{lem:G-updates}} \label{app:G-updates}
\begin{lemma*}
Let $\cG(\beta, t) = \{u\in\cS(1): \PP[\|\tilde{A} u\|_\infty \leq \frac{1}{d^{4.5}}] \geq \beta \}$, where the probability is taken with respect to the conditional distribution of $A$ given the realizations of $B$, $q$, $u_{\leq t}$, and $\hat{a}_{\leq t}$. Then for any $\beta_0\in[0,1]$, 
with probability at least $1-\frac{1}{d^3}$ (w.r.t. $A$) it holds $\forall t\in [T]$ that
$\cG(2^{4t\log(d)}\beta_0, t) \subseteq \cG(\beta_0, 0).$
\end{lemma*}
\begin{proof}
We will prove the result by induction. Specifically, we will show that if $\cG(2^{4(t-1)\log(d)}\beta_0, t-1) \subseteq \cG(\beta, 0)$, then 
$\pr{}{\cG(2^{4t\log(d)}\beta_0, t) \subseteq \cG(\beta, 0)} \geq 1-\frac{1}{d^4}$. The base case for $\cG(\beta,0)$ clearly holds.

For any $t\in[T]$, denote $\tau_t=\ip{\hat{a}_t}{u_t} - \hat{\eta}_{t,i_t}$ (recall $i_t$ is the index of the row returned from the query). %
Let $\tilde{\eta}_t$ denote the values of $\eta_t$ associated with the rows of $\tilde{A}$ (in the corresponding order).
Let $\cH_t$ be the $\sigma$-algebra generated by the randomness in Algorithm \ref{alg:MSGH} up until the query is made at round $t$. That is, the randomness in $\{B,q,u_{1:t},\eta_{1:t},\tau_{1:t-1}, \hat{a}_{1:t-1}\}$. %

Now let $\cE$ be some event on $\tilde{A}$.
Let $\cA_t(\tau)$ denote the set of all possible realizations of $\tilde{A}$ consistent with $\cH_t$ that result in $\tau_t = \tau$.
Observe that since $\hat{a}_t$ is never a row of $\tilde{A}$, any realization of $\tilde{A}$ satisfying $\|\tilde{A} u_t - \tilde{\eta}_t\|_\infty \leq \tau_t$ leads to the same values of $\hat{a}_t$ and $\tau_t$. 
By Bayes rule, the posterior distribution of a random variable conditioned on the event that the random variable is in some set is just a renormalization over that set. Thus, the posterior distribution of $\tilde{A}$ after receiving $\hat{a}_t$ and $\tau_t$ is obtained by a renormalization over the set $\cA_t(\tau_t)$, which does not depend on $\hat{a}_t$; consequently the event $\cE$ is conditionally independent of $\hat{a}_t$ given $\tau_t$ and $\cH_t$.
This also implies the event is conditionally independent of $u_{t+1}$ and $\eta_{t+1}$, as these are post processing. 
Thus $\PP[\cE| \cH_{t+1}] = \PP[\cE| \tau_t, \cH_{t}] 
    = \PP[\cE | \tilde{A} \in \cA_t(\tau_t), \cH_t]$.
We then have,
\begin{align}\label{eq:event-increase}
    \PP[\cE| \cH_{t+1}] = \PP[\cE | \tilde{A} \in \cA_t(\tau_t), \cH_t] = \frac{\PP[\tilde{A} \in \cA_t(\tau_t)~|~ \cE,\cH_t]\PP[\cE ~|~ \cH_t]}{\PP[\tilde{A} \in \cA_t(\tau_t) | \cH_t]} 
    \leq \frac{\PP[\cE ~|~ \cH_t]}{\PP[\tilde{A} \in \cA_t(\tau_t) ~|~ \cH_t]}.
\end{align}%
We would thus like to lower bound $\PP[\tilde{A} \in \cA_t(\tau_t) ~|~ \cH_t]$.
Now consider $\tau_t^*\in\re$ as the largest value such that $\prnos{}{\tilde{A}\in\cA_t(\tau_t^*)} = \frac{1}{d^4}$.
For any $\tau_1 < \tau_2$, $\cA_t(\tau_1) \subseteq \cA_t(\tau_2)$. 
Thus so long as $\tau_t \geq \tau_t^*$, we have
$\PP[\tilde{A} \in \cA_t(\tau_t) | \cH_t] \geq \frac{1}{d^4}$.
Since $\tau \leq \tau^*$ implies $\tilde{A} \in \cA_t(\tau^*)$, $\PP[\tau_t \leq \tau^*] \leq \PP[\tilde{A}\in\cA_t(\tau^*)] \leq \frac{1}{d^4}$.
Thus we have $\PP[\tau_t \geq \tau^* ~|~ \cH_t] \geq 1-\frac{1}{d^4}$, and then under this event, by Eqn. \eqref{eq:event-increase}
no event on $\tilde{A}$ increases in probability by more than a factor of $d^4$.
Thus with probability at least $1-\frac{1}{d^4}$, for every $v \notin \cG_B$,
\begin{align*}
    \pr{}{\|\tilde{A}v\|_\infty \leq \frac{1}{d^{4.5}} ~|~ \cH_{t+1}} \leq d^4\cdot \pr{}{\|\tilde{A}v\|_\infty \leq \frac{1}{d^{4.5}} ~|~ \cH_{t}} \leq 2^{4(t+1)\log(d)}\beta_0. %
\end{align*}
This establishes the induction hypothesis. 
Now since $T \leq \frac{d}{C_1\log(d)} < d$, by the induction argument and a union bound the lemma statement follows.
\end{proof}

\section{Supplement to Section \ref{sec:deterministic}}
Some of the proofs in this section will use the following fact.
\begin{lemma}\label{lem:marginal-unif-vec}
Let $j\in[N],l\in[L]$. Fix any realization of $\cL_{j,l}$. The conditional distribution of $\hat{x}_{j,l}$ is uniform over $\cL_{j,l}^\perp \cap \cS(1)$. 
\end{lemma}
\begin{proof}
This follows from the fact that the orthogonal complement of $\cL_{j,l}$ in $\cK_{j,l}$, $\cK_{j,l}\cap\cL_{j,l}^\perp$, is a uniformly random linear subspace over $\cL_{j,l}^\perp$ of dimension at least $1$. 
Thus, sampling a unit vector uniformly at random in $\cK_{j,l}\cap\cL_{j,l}^\perp$ results in a uniformly random unit vector over $\cL_{j,l}^\perp$ after marginalizing over $\cK_{j,l}$.
\end{proof}

\subsection{Proof of Lemma \ref{lem:detalg-opt-ocvg}} \label{app:detalg-opt-ocvg}
\begin{lemma*}
For any $j\in[N-1], l\in[L]$. 
There exists a strategy for choosing $\cL_{j,l}' \in \Gr([2k],d)$ such that,
\begin{align*}
\underset{}{\PP}\Big[\forall t\in[T]: ~\|A\ujlt\|_\infty \geq \frac{\|\ujlt\|}{2d^{4.5}} ~\text{ or }~  \frac{\ujlt}{\|\ujlt\|} \in \cL_{j,l}' \msum\cB\Big(\frac{1}{d^{1.5}}\Big)\Big] = 1-O(1/d^3).
\end{align*}
The probability is taken with respect to all randomness in Algorithm \ref{alg:detalg-opt}.
\end{lemma*} 
\begin{proof}
Define $B_{j,l}$ as the latest $m$-bit memory state of $\opt$ before the first oracle query of round $(j,l)$ is made. We denote the set of all previous round indices as $\cI = ([j]\times [L]) \cup (\{j\} \times [l-1])$. 

We will show that if Eqn. \eqref{eq:detalg-opt-ocvg} were not satisfied, then there exists a player strategy which wins the MSGH with non-negligible probability. 
Specifically, we consider a player, $\cP,$ playing the MSGH with embedding dimension $d$, matrix size $d'=d/2$, query limit $\frac{d}{C_1\log(d)}$, and message size $m_1 = 2m$, hint size $m_2 = \frac{d}{C_2}$, 
dimension threshold $2k$, and norm parameters $\Delta_1,...,\Delta_{d'}=1$.

The player will implement a strategy for the MSGH by using $\opt$ and Algorithm \ref{alg:detalg-opt}. %
Consider running Algorithm \ref{alg:detalg-opt} with $\opt$ until round $(j,l)$, where the matrix $A$ from Algorithm \ref{alg:detalg-opt} is set to be the MSGH matrix, which we note has the same distribution. 
For the purposes of proving the result of the lemma for round $(j,l)$, the process for choosing the subspaces $\{\cL'_{j',l'}\}_{(j',l')\in\cI}$ does not matter, so long as they each satisfy the dimension bound.
Note that running $\opt$ up to round $(j,l)$ also samples the Nemirovksi vectors $\{x_{j',l'}\}_{(j',l')\in\cI}$.

The player chooses the function $h_1$ to be set such that $h_1(A)$ outputs $B_{j,l}$ (computed by running $\opt$ as previously described) concatenated with the bit representations of the Nemirovski vectors, $\{x_{j',l'}\}_{(j',l')\in\cI}$, as determined by their corresponding covers. 
Note we can assume $\cP$ retains access to the quantities $\{\cK_{j,l}\}_{j\in[N],l\in[L]}$ and their associated covers, since they are independent of $A$ (i.e. they can be sampled before $h_1$ is computed). 
To show this implementation of $h_1$ is possible under the setting $m_1 = 2m$, we consider the number of bits needed to the store the vectors.
Let $(j',l')\in\cI$. Since $\cV_{j',l'}$ is a minimal $(d^{-1.5})$-cover of a $3kl$ dimensional ball, we have $\log(|\cV_{j',l'}|) \leq 3kl\log(3d^{1.5})$ (Lemma \ref{lem:packing-ball}). Further, since $\cV_{j',l'}$ is independent of $A$, we can store $x_{j',l'}$ using at most $3kl'\log(3d^{1.5})$ bits that depend on $A$.
Thus we can store up to $NL$ such vectors with $3kNL^2\log(3d^{1.5}) \leq 3k(k^{1/3})(\frac{d}{20k\log(d)})^2\log(3d^{1.5}) \leq (\frac{dk}{32}) \frac{d}{k^{5/3}} \leq \frac{dk}{32} = m$ bits (recall we have assumed $k > d^{3/5}$).
Thus the overall message size is at most $2m = m_1$ bits as required. 

We now set $\cL_{j,l}'$ from Algorithm \ref{alg:detalg-opt} to be the 
subspace $\cL$ that is marked by the adversary in the MSGH, where we assume the adversary uses any strategy that realizes Proposition \ref{prop:msgh-bound}.

For the hint, $\cP$ takes $h_2$ to output the discretized Nemirovski vector of the current round, $h_2(\cL,A)=x_{j,l}$. Again we note the vector can be specified using at most $2kl'\log(3d^{1.5}) \leq \frac{d}{C_2}$ bits that depend on $A$ (recall $L=\frac{d}{32C_2k\log(d)}$).
At this point, the only information $\cP$ has about $A$ is $h_1(A)$ and $h_2(\cL',A)$. Note $\cP$ still has access to the subspaces $\cK_{j,l}$ and their respective covers, as these are sampled independently of $A$ and $\cL_{j,l}'$.

Now the $\cP$ continues running the optimization process through the inner loop from $t=1...T$ as per Algorithm \ref{alg:detalg-opt}, which $\cP$ can do because it has stored the memory state of $\opt$ right before the start of round $(j,l)$ and, as we will show subsequently, $\cP$ can compute the answers to the resisting oracle queries. 
Each time $\opt$ makes an oracle query $\ujlt$, $\cP$ makes queries the MSGH with teh same vector, $\ujlt$, and evaluates $\tilde{\cO}_{j,l}(\ujlt)$ using the Nemirovski vectors and the row of $A$ it receives from the MSGH; $\cP$ then gives $\tilde{\cO}_{j,l}(\ujlt)$ to $\opt$. We assume here the player always takes all the query offsets $\eta_{1:T,1:d'}$ to be zero so that the row received from $A$ matches the row which would be the subgradient of the barrier function, $\barr$. Since the resisting oracle responses only depends on $\{x_{j',l'}\}_{(j',l')\in\cI} \cup \{x_{j,l}\}$ and the barrier function, the player is always able to evaluate the oracle correctly.

Now since the optimization problem uses the same matrix $A$, as the MSGH, and $\cL_{j,l}$ is the marked subspace, it is clear that if $\opt$ makes a query $\ujlt$ such that
$\|A\ujlt\|_\infty \geq \frac{\|\ujlt\|}{2d^{4.5}} ~\text{ or }~  \frac{\ujlt}{\|\ujlt\|} \in \cL_{j,l}' \msum\cB(1/d^{1.5})$,
then the player wins the MSGH. Thus Proposition \ref{prop:msgh-bound} guarantees there exists a strategy for choosing $\cL_{j,l}'$ such that with probability at least $1-O(1/d^3)$ every $t\in[T]$ satisfies 
$\|A\ujlt\|_\infty \geq \frac{\|\ujlt\|}{2d^{4.5}} ~\text{ or }~  \frac{\ujlt}{\|\ujlt\|} \in \cL_{j,l}' \msum\cB(1/d^{1.5})$.
\end{proof}

\subsection{Proof of Lemma \ref{lem:detalg-resisting-consistent}}\label{app:detalg-resisting-consistent}
\begin{lemma*}
For Algorithm \ref{alg:detalg-opt}, $\pr{}{\forall (j,l,t)\in[N]\times[L]\times[T]: \tilde{\cO}_{j,l}(\ujt) = \cO(\ujt)} \geq 1-O(\frac{1}{d})$, where the probability is taken w.r.t. the randomness in Algorithm \ref{alg:detalg-opt}.
\end{lemma*}
\begin{proof}
In the following we condition on two events happening.
First, we assume the event specified by Lemma \ref{lem:detalg-opt-ocvg} holds for every round, which happens with probability at least $1-O(NL/d^3)$.
Second, we condition on the event that 
for any round $(j,l)$ and any $j' > j$ and $l'\in [L]$ that
$|\ip{\ujlt}{\hat{x}_{j',l'}}| \leq 5\sqrt{\frac{\log(d)}{d}}$. 
For the latter event, observe that the marginal distribution of $\hat{x}_{j',l'}$ is uniformly random over $\cL_{j',l'}^\perp$ (Lemma \ref{lem:marginal-unif-vec}).
Since $\Dim(\cL_{j',l'}^\perp) \geq d-3Lk \geq d/2$, we have 
$\pr{}{|\ip{\ujlt}{\hat{x}_{j',l'}}| \geq 4\sqrt{\frac{\log(d)}{d}}} \leq \frac{1}{d^4}$ by a standard concentration inequality.
The probability that this holds for all $j'>j$ and $l'\in[L]$ for every query at every round is at least $1-(NL)^2d/d^4 \geq 1-1/d$. Thus this event and the event of specified by Lemma \ref{lem:detalg-opt-ocvg} both happen with probability at least $1-O(1/d)$ by a union bound.

What remains is to show consistency of the resisting oracle under the above event. We show this by considering 3 different cases. 
\paragraph{Case 1: $\|\ujlt\| \in [0,\frac{1}{d^{1.5}}]$.}
In this case, the gradient can only be either the gradient of the barrier function or the first Nemirovski vector, $x_{1,1}$.

\paragraph{Case 2: $\|\ujlt\| \in [\frac{1}{d^{1.5}},1] ~~\&~~ \|A\ujlt\|_\infty \geq \frac{\|\ujlt\|}{2d^{4.5}}$.}
By the definition of $\barr$ and the fact that $\|\ujlt\| \geq d^{-1.5}$ we have $\barr(\ujlt) = 4d^6\|Aw\|_\infty-\rho \geq 2 - \rho \geq 1.5$, but $\nem(\ujlt) \leq 1$ for any $\ujlt \in \cB(1)$. Thus the gradient is returned by the barrier function which is the same for both $F_{j,l}$ and $F$. 

\paragraph{Case 3: $\|\ujlt\| \in [\frac{1}{d^{1.5}},1] ~~\&~~ \|A\ujlt\|_\infty \leq \frac{\|\ujlt\|}{2d^{4.5}}$. } 
Recall we have conditioned on the case where the event specified by Lemma \ref{lem:detalg-opt-ocvg} holds for every round and query; see specifically Eqn \eqref{eq:detalg-opt-ocvg}. That is, because $\|A\ujlt\|_\infty \leq \frac{\|\ujlt\|}{2d^{4.5}}$, it must be that
$\frac{\ujlt}{\|\ujlt\|} \in \cL_{j,l}' \msum \cB(1/d^{1.5})$.
Since each $\hat{x}_{j,i}$, $i \geq l$, is sampled orthogonal to $\cL_{j,l}'$ (recall $\cL_{j,l}' \subseteq \cL_{j,i})$, we have that 
$|\ip{\ujlt}{x_{j,i}}| \leq \|\proj_{\cL_{j,i}}^\perp  \ujlt\| + d^{-1.5} \leq \frac{2}{d^{1.5}}$. The scale of $\gamma_2=5d^{-1.5}$ then implies the subgradient of $\nem(\ujlt)$ does not contain $x_{j,i}$ for any $i>l$. 

We now want to show that the gradient is not $x_{j',l'}$ for any $j' > j$ and $l'\in[L]$. 
Recall we have conditioned on the event where 
$|\ip{\ujlt}{\hat{x}_{j',l'}}| \leq 4\sqrt{\frac{\log(d)}{d}}$.
Observe $|\ip{\ujlt}{x_{j',l'}}|$ can only be $d^{-1.5}$ larger since $\cV_{j',l'}$ is a $(d^{-1.5})$-cover.
Finally, as noted in the previous paragraph, $|\ip{\ujlt}{x_{j,l}}| \leq 2d^{-1.5}$. Thus, recalling $\gamma_1 = 11\sqrt{\frac{\log(d)}{d}}$ and $\gamma_2 = \frac{5}{d^{1.5}}$,
\begin{align*}
    \ip{\ujlt}{x_{j,l}} - j\gamma_1 - l\gamma_2 
    &\geq -\frac{2}{d^{1.5}} - j\gamma_1 - \frac{5L}{d^{1.5}} \\
    &> -j\gamma_1 -  \frac{1}{4\sqrt{d}}\\
    &\geq -(j+1)\gamma_1 -l'\gamma_2 - \ip{\ujlt}{x_{j',l'}}. 
\end{align*}
This implies no $x_{j',l'}$ is in the subgradient of $\nem(\ujlt)$, and we are done.
\end{proof}

\subsection{Proof of Lemma \ref{lem:detalg-minimizer}} \label{app:detalg-minimizer}

\begin{lemma*}
With probability at least $1-2^{-\Omega(d)}$ (with respect to all randomness in Algorithm \ref{alg:detalg-opt}) there exists $w^*\in\cB(1)$ such that 
$F(w^*) \leq -\frac{\sqrt{\log(d)}k^{1/3}}{\sqrt{d}}$.
\end{lemma*}
\begin{proof}
We will iteratively construct a series of orthogonal vectors, $\tilde{x}_{1:N,1:L}$.
For $(j,l)\in[N]\times[L]$, let $\tilde{x}_{j,l} = \proj_{A,\cX_{j,l}}^\perp \hat{x}_{j,l}$, where $\cX_{j,l}=\Span(\{\tilde{x}_{j',l'}: (j',l')\in ([j-1]\times[L]) \cup (\{j\}\times[l-1])\}$ (i.e. the span of the previously constructed vectors). 
In the following, condition on any realization of $A$, $\bar{A}$, and $\cL_{1:N,1:L}$, $\bar{\cL}_{1:N,1:L}$, in their support. 

Consider some $j\in[N],l\in[L]$. Since we have conditioned $\cL_{j,l}=\bar{\cL}_{j,l}$, the vector $\hat{x}_{j,l}$ is a uniformly random unit vector over $\bar{\cL}_{j,l}^\perp$ (Lemma \ref{lem:marginal-unif-vec}).
Let $y\sim\Unif(\cS(1))$ and consider the coupling where $y = a\hat{x} + v$ for $a\in[0,1]$ and $v\in\bar{\cL}_{j,l}$ with appropriate distributions. 
Now,
\begin{align*}
\|\tilde{x}_{j,l}\|=\|\proj_{\bar{A},\cX_{j,l}}^\perp \hat{x}_{j,l}\| \geq \|\proj_{\bar{A},\cX_{j,l},\bar{\cL}_{i,j}}^\perp \hat{x}_{j,l}\| \geq \|\proj_{\bar{A},\cX_{j,l},\bar{\cL}_{i,j}}^\perp a\hat{x}_{j,l}\| =\|\proj_{\bar{A},\cX_{j,l},\bar{\cL}_{i,j}}^\perp y\|.
\end{align*}
At the same time, since $\Dim(\Span(\bar{A},\cX_{j,l},\bar{\cL}_{i,j})) \leq \frac{d}{2} + \frac{d}{8}  + \frac{d}{8} \leq \frac{3d}{4}$, we have,
\begin{align*}
    \pr{\cK_{j,l},\hat{x}_{j,l}}{\|\proj_{\bar{A},\cX_{j,l}}^\perp \hat{x}_{j,l}\| \geq \frac{1}{4}} \geq \pr{y_{j,l}}{\|\proj_{\bar{A},\cX_{j,l},\bar{\cL}_{j,l}}^\perp y_{j,l}\| \geq \frac{1}{4}} \geq 1 - 2^{-\Omega(d)}. 
\end{align*}

Under this event we have, $ \ip{\tilde{x}_{j,l}}{\hat{x}_{j,l}} = \|\tilde{x}\|^2 \geq \frac{1}{16}$.
Since $\|\hat{x}_{j,l} - x_{j,l}\| \leq d^{-1.5}$, we have for $d$ sufficiently larger than some constant that, 
$\ip{\tilde{x}_{j,l}}{x_{j,l}} \geq \frac{1}{32}$. 
Via a simpler analog of the above it is also straightforward to verify that 
$\pr{}{\ip{\tilde{x}_{N,l}}{x_{N,l}} \geq \frac{1}{32}} \geq 1-2^{-\Omega(d)}$ for any $l\in[L]$.

Thus, this inner product lower bound holds for all $(i,j)\in[N]\times [L]$ with probability at least $1-NL 2^{-\Omega(d)} \geq 1-2^{-cd}$, for some constant $c$ and $d$ larger than some constant.
We have lower bounded the conditional probability of this event for any realization of $A$ and $\cL_{1:N,1:L}$, thus the same bound holds for the unconditional probability.

Now we take $w^* = -\frac{1}{\sqrt{NL}}\sum_{j=1}^N \sum_{l=1}^L \tilde{x}_{j,l}$. Since each $\tilde{x}_{j,l}$ is orthogonal to $A$ and $\tilde{x}_{1:N,1:L}$ is an orthogonal set of vectors, this yields
$F(w^*) \leq \max\{-\frac{\sqrt{\log(d)}k^{1/3}}{\sqrt{d}}, -\rho\} = -\frac{\sqrt{\log(d)}k^{1/3}}{\sqrt{d}}$ provided the previously described high probability events hold for every $j,l$. 
It is straightforward to verify that $w^*\in\cB(1)$.
\end{proof}

\section{Supplement to Section \ref{sec:randomized}}

\subsection{Proof of Lemma \ref{lem:minimizer}} \label{app:minimizer}
\begin{lemma*} 
With probability at least $1-2^{-\Omega(d)}$ (under the draw of $F$) there exists $w^*\in\cB(1)$ such that $F(w^*) = 0$.   
\end{lemma*}
\begin{proof}
Consider $w^*= \sum_{j=1}^{N}s_j\tilde{x}_j$ for some $s\in\re^N$. 
Observe,
\begin{align*}
    \ip{w^*}{x_j} &= \ip{s_j \tilde{x}_j}{x_j} + \sum_{i\in [N]\setminus j} \ip{s_i \tilde{x}_i}{x_j} \\
    &= s_j\ip{\tilde{x}_j}{\proj_A x_j + \tilde{x}_j} + \sum_{i\in [N]\setminus j} \ip{s_i \tilde{x}_i}{\proj_A x_j + \tilde{x}_j} \\
    &= s_j\|\tilde{x}_j\|^2.
\end{align*}
Thus setting $s_j = \frac{2\alpha}{\|\tilde{x}_j\|^2}$ for each $j\in[N]$ ensures that $\max\{\nem(w^*),\barr(w^*)\}=0$.

Before handling the wall function, to verify that $w^* \in \cB(1)$, observe that with probability at least $1-2^{-\Omega(d)}$, each $\|\tilde{x}_j\| \geq \frac{1}{2}\|x_j\| \geq 1/2$, by the properties of random projection (Lemma \ref{lem:random-projection}). Thus under this event we have $\|s\|_\infty \leq 8\alpha$ and $\|w^*\| \leq 8\alpha\sqrt{N} \leq 1$, where the last inequality holds because $\alpha \leq \frac{1}{\sqrt{N}} = \frac{k^{1/4}}{\sqrt{8}d^{1/4}}$.

Finally, for the wall function component, observe that under the event that $\|s\|_\infty \leq 8\alpha$, we have for any $z\in\cZ$,
\begin{align*}
    |\ip{w^*}{z}| \leq 8\alpha \sum_{j=1}^N |\ip{\tilde{x}_j}{z}| \leq 48\alpha N \sqrt{\frac{k}{d}} \leq \alpha. 
\end{align*}
Thus at $w^*$, the wall function has value $\wall(w^*)\leq-3\alpha$ and thus $F(w^*) = 0$.
\end{proof}

\subsection{Proof of Lemma \ref{lem:opt-ocvg}} \label{app:opt-ocvg}
In the following, recall that for any $j\in[N]$, $\tujt := \proj_{\tilde{x}_{<j}}\ujt$.
\begin{lemma*}
Consider Algorithm \ref{alg:opt}. Fix $j\in[N-1]$. 
It holds that,
\begin{align} \label{eq:opt-ocvg}
\underset{A,x_{\leq j}}{\PP}\Big[\exists t\in[T]: \! ~\|A\tujt\|_\infty \leq \frac{\|\tujt\|}{8d^{5.5}} ~\text{ and }~  \max\{|\ip{\tujt}{x_j}|,|\ip{\tujt}{\tilde{x}_j}|\} \geq 6\|\tujt\|\sqrt{\frac{k}{d}}\Big] = O\big(\frac{1}{d^3}\big).
\end{align}
\textit{(note also $\|A\ujt\|_\infty = \|A\tujt\|_\infty$ and $\ip{\ujt}{x_j} = \ip{\tujt}{x_j}$.)} 
\end{lemma*} 
\begin{proof}
In the following, for any $l\in[N]$ define $\bar{x}_l = \proj_A x_l$. Note $x_l = \bar{x}_l + \tilde{x}_l$. 
Further define $\bar{a}_i = \proj_{\bar{x}_{<j}} a_i, ~\forall i\in[d']$ . 

We will show this by showing that if Eqn. \eqref{eq:opt-ocvg} were not satisfied, then there exists a strategy which wins an OCVG (see Proposition \ref{prop:ocvg-bound}) with non-negligible probability. Several sublemmas used in the following are deferred to the next section, Appendix \ref{app:opt-ocvg-sublemmas}.

\paragraph{Sampling a loss instance using the OCVG.} First, the player $\cP$ samples $\tilde{x}_{< j}$, $x_{\leq j}$, and $\bar{a}_{1:n}$ from their joint  distribution (note $\bar{x}_{<j}$ is determined by $\tilde{x}_{<j}$ and $x_{<j}$), as induced by the sampling process specified at the start of Section \ref{sec:proof-of-main-lb}. 
$\cP$ then receives an OCVG instance with embedding dimension $d-j+1$, matrix size $d'$, query limit $T=\frac{d}{C_1\log(d)}$, message sizes $m_1=m$ and $m_2=0$, and dimension threshold $k=32\frac{m}{d}$ (although note that in the OCVG this parameter acts more as a correlation threshold). 
The norm parameters, $\Delta_1,...,\Delta_{d'}$ in the OCVG are set such that $\Delta_{i} = \sqrt{1- \|\bar{a}_i\|^2}$.
Lemma \ref{lem:proj-norms}, deferred to Section \ref{app:opt-ocvg-sublemmas} below, shows that w.p. at least $1-2^{-\Omega(d)}$, $\Delta_i \geq 1/2$ for all $i\in[d']$, as required.
For clarity, we call the uniformly random unit vector from the OCVG $x'$ and the matrix from the OCVG $A'$ (whose $i$'th row we denote as $a'_i$).

$\cP$ will implement its strategy by using $\opt$ to induce some $h_1$ and compute the informative bits $h_1(A')$, and then continue using $\opt$ to make the OCVG queries, but
we first show how $\cP$ can sample an optimization problem from the distribution defined at the start of Section \ref{sec:proof-of-main-lb} using the OCVG quantities $A'$ and $x'$. 
Let $M$ be any $d\times (d-j+1)$ orthonormal matrix whose column span is $\Span^\perp(\tilde{x}_{<j})$. 
Each row of $A$ is set as $a_i = \bar{a}_i + Ma'_i$. 
Due to the sampling distribution of the OCVG matrix $A'$, each $M a_i'$ is uniformly random over $\Span^\perp(\tilde{x}_{<j}) \cap \cS(\Delta_i)$, and thus setting $A$ this way induces the same distribution over $A$ that would result from the sampling defined at the start of Section \ref{sec:proof-of-main-lb}; this is shown by Lemma \ref{lem:conditional-of-rows}, deferred to Section \ref{app:opt-ocvg-sublemmas} below. $\cP$ takes $x_j = Mx'$, and we thus have $x_j \sim \Unif(\{\Span^\perp(\tilde{x}_{<j})\cap \cS(1)\}$ since $x'$ is a uniformly random $d-j+1$ dimensional unit vector.
The other loss function parameters $k$ and $\alpha$ are set as described at the start of Section \ref{sec:proof-of-main-lb}. Note that because we will only consider running the optimization process (Algorithm \ref{alg:opt}) until the end of round $j$, and we use a resisting oracle, the values of $x_{j+1:N}$ and $\tilde{x}_{j:N}$ will not be needed by the $\cP$ or $\opt$.

\paragraph{Using $\opt$ to implement the player strategy.} 
We can now describe the player strategy.
$\cP$ takes $h_1(A')$ to be the (randomized) function which first samples the objective function according the previously described procedure, then runs $\opt$ until the start of the $j$'th round and outputs the last $m$-bit memory state of $\opt$ before the first oracle query of round $j$ is made.
Note that while $h_1(A')$ is the only information $\cP$ has about $A'$, $\cP$ still has access to $\tilde{x}_{<j}$, $x_{\leq j}$, and $\bar{a}_{1:n}$, as these were sampled independently of $A'$.

$\cP$ now implements the ``querying'' phase of the OCVG by continuing to run $\opt$ until the end of the $j$'th iteration and making OCVG queries that are transformations of the oracle queries made by $\opt$.
This is possible because $\cP$ has access to memory state of $\opt$ right before the start of round $j$ and, as we will show subsequently, is able to answer the resisting oracle responses.
Specifically, at each resisting oracle query $\ujt$, $\cP$ makes the OCVG query $M^{-1}\ujt = M^{-1}\tujt$, and for all $i\in[d']$ sets $\eta_{t,i} = -\ipnos{\ujt}{\bar{a}_i}$ (recall $\cP$ has access to $\bar{a}_{1:d'}$). 
From the OCVG query $\cP$ receives the vector $a'_{i^*}$ where,
\begin{align*}
    i^{*} = \argmax\limits_{i\in[d']}\{\ip{M^{-1}\ujt}{a'_i} - \eta_i\} =  \argmax\limits_{i\in[d']}\{\ip{\ujt}{M a'_i} + \ipnos{\ujt }{\bar{a}_i}\} = \argmax\limits_{i\in[d']}\{\ip{\ujt}{a_i}\}.
\end{align*}
Since $\cP$ also has access to $x_{\leq j}$ and $\tilde{x}_{<j}$, this ensures $\cP$ has all the information needed to evaluate the response of the resisting oracle, which it can then send to $\opt$.
Specifically, if $\nabla F_j(\ujt) = \nabla \barr(\ujt)$, $\cP$ sends $\opt$ the vector $d^{6} (Ma'_{i^*}+\bar{a}_{i^*})=d^{6}a_{i^*} = \nabla\cR(\ujt)$. Otherwise it implements the resisting oracle using $\tilde{x}_{<j}$ and $x_{\leq j}$.

\paragraph{Showing an informative oracle query wins the OCVG.} What remains is to show that if, 
\begin{align}\label{eq:opt-ocvg-win-event}
    \exists t\in[T]: \|A\tujt\|_\infty \leq \frac{\|\tujt\|}{8d^{5.5}} \text{ and } \max\{|\ip{\tujt}{x_j}|,|\ip{\tujt}{\tilde{x}_j}|\} \geq 6\|\tujt\|\sqrt{\frac{k}{d}}.
\end{align}
then $\cP$ wins the OCVG.

For any $i\in[d']$, the $i$'th row of $MA'$ is equal to $\proj_{\bar{x}_{<j}}^\perp a_i$ which we note is in $\Span(A)$ since $\Span(\bar{x}_{<j}) \subseteq \Span(A)$. Consequently, by Lemma \ref{lem:Aprime-lb-A}, deferred to Section \ref{app:opt-ocvg-sublemmas} below, we have with probability at least $2^{-\Omega(d)}$ under the draw of $A$, 
\begin{align} \label{eq:nice-embedding}
    \forall v\in\cB(1), \quad \|M A' v\|_\infty \leq 4d\|A v\|_\infty \quad \text{ and }\quad |\ip{v}{x_j} - \ip{v}{\tilde{x}_j}| \leq 2d\|Av\|_\infty.
\end{align}
Using this, we see that if Eqn. \eqref{eq:opt-ocvg-win-event} is satisfied for some $t\in[T]$, then
$\|A' M^{-1} \tujt\|_\infty = \|M A' \tujt\|_\infty \leq  (4d)\frac{\|\tujt\|}{8d^{5.5}} = \frac{\|\tujt\|}{2d^{4.5}} = \frac{\|M^{-1}\tujt\|}{2d^{4.5}}$, which ensures the OCVG query satisfies near orthogonality.
For the sufficient correlation condition, first consider the case where 
$|\ip{\tujt}{x_j}| \geq 6\|\tujt\|\sqrt{\frac{k}{d}}$.
Then since $M$ is orthonormal and $\tujt$ lies in its column span, the correlation condition is satisfied as,
\begin{align*}
    |\ip{M^{-1}\tujt}{x'}| = |\ip{\tujt}{x_j}| \geq 6\|\tujt\|\sqrt{\frac{k}{d}} = 6\|M^{-1}\tujt\|\sqrt{\frac{k}{d}}. %
\end{align*}
Alternatively, in the case where $|\ip{\tujt}{\tilde{x}_j}| \geq 6\|\tujt\|\sqrt{\frac{k}{d}}$, then by Eqn. \eqref{eq:nice-embedding} we still have 
$|\ip{\tujt}{x_j}| \geq 6\|\tujt\|\sqrt{\frac{k}{d}} - 2d\|A\tujt\|_\infty \geq 6\|\tujt\|\sqrt{\frac{k}{d}} - \frac{\|\tujt\|}{4d^{5.5}} \geq 4\|\tujt\|\sqrt{\frac{k}{d}}$, and we conclude as in the first case.
Finally, $m_1 = m \leq \frac{d^2}{256\log(d)} \leq \frac{(d-j+1)^2}{64\log(d-j+1)}$ since $j < N < \sqrt{d}/48$.
Thus the probability of $\cP$ winning is at most $O(1/d^3)$ by  Proposition \ref{prop:ocvg-bound}. Since we have shown that $\cP$ wins the OCVG when Eqn. \eqref{eq:opt-ocvg-win-event} is satisfied, this yields the claim of the lemma.
\end{proof}

\subsubsection{Sublemmas for the proof of Lemma \ref{lem:opt-ocvg}}\label{app:opt-ocvg-sublemmas}
The following lemmas are used in the proof of Lemma \ref{lem:opt-ocvg}.
\begin{lemma} \label{lem:proj-norms}
Let $j\in[N]$ and $d$ be larger than some constant. Then $\pr{}{\forall i\in[d']: \|\proj_{\bar{x}_{<j}} a_i\| \geq 1/2} \geq 1-2^{-\Omega(d)}$.

\end{lemma}
\begin{proof}
For any $\bsA \in\Supp(A)$, condition on $A=\bsA$. Since each $x_{j'}$, $j'<j$, is sampled uniformly at random over a linear subspace containing $\Span(\bsA)$ (the row span of $\bsA$), the conditional distribution given $A=\bsA$ of the vector $\frac{\bar{x}_{j'}}{\|\bar{x}_{j'}\|}$ is uniform over $\Span(\bsA)\cap\cS(1)$; this follows from the rotational invariance of the distribution of $x_{j'}$. Thus, $\Span(\bar{x}_1,...,\bar{x}_{j-1})$ is a uniformly random linaer subspace inside $\Span(\bsA)$. 
Since $j<N<d/48$, by the properties of random projection, for any $i\in[d']$,
\begin{align*}
    \pr{x_{<j}|A=\bsA}{\|\proj_{\bar{x}_{<j}}\bsa_i\| \geq \frac{1}{2}} \leq \pr{}{\|\proj_{\bar{x}_{<j}}\bsa_i\| \geq 2\sqrt{\frac{j}{d}}} \leq 2^{-\Omega(d)}.
\end{align*}
Marginalizing over $A$ yields the same bound on the unconditional probability.
The final result follows from a union bound and the fact that $d$ is sufficiently larger than a constant.
\end{proof}

For the following, we recall that for each $l\in[N]$, $\bar{x}_l = \proj_A x_l$. Further, for the value of $j$ fixed in Lemma \ref{lem:opt-ocvg}, we have defined $\bar{a}_i = \proj_{\bar{x}_{<j}} a_i$. We give an explicit form for the conditional distribution of $A$ given certain fixed information. The proof ideas are similar in nature to the proof of \cite[Lemma 3]{MN25}.
\begin{lemma} \label{lem:conditional-of-rows}
Fix any realization of $\tilde{x}_{< j}$,$\bar{x}_{<j}$,$x_{\leq j}$, and $\{\bar{a}_i\}_{i\in[d']}$. 
Let, 
\begin{align*}
    \cA_i = \bc{a\in\cS(1) : \proj_{\bar{x}_{<j}}a = \bar{a}_i ~~\text{ and }~~ \proj_{\bar{x}_{<j}}^\perp a \in \Span^\perp(\tilde{x}_{<j}) \cap \cB(\sqrt{1-\|\bar{a}_i\|^2})},
\end{align*}
and let $\cA\subseteq \Supp(A)$ be the set of matrices such that $\hat{A} \in \cA$ if each row $\hat{a}_i$, $i\in [d']$, satisfies $\hat{a}_i \in \cA_i$.
The conditional distribution of $A$ given $\tilde{x}_{< j}$, $\bar{x}_{<j}$ $x_{\leq j}$, and $\{\bar{a}_i\}_{i\in[d']}$ is uniform over $\cA$. 
\end{lemma}
\begin{proof}
Since $\bar{x}_l = x_l - \tilde{x}_l$ for all $l\in[N]$, it suffices to show the result when conditioning on $\tilde{x}_{< j}$,$x_{\leq j}$, and $\{\bar{a}_i\}_{i\in[d']}$

Let $\bar{A}$ be the matrix whose $i$'th row is $\bar{a}_i$ for each $i\in[d']$. Consider any $\bsA\in\Supp(A)$, $\bsbarA\in\Supp(\bar{A})$, $\bsx_{\leq j} \in \Supp(x_{\leq j})$ and $\tilde{\bsx}_{<j}\in\Supp(\tilde{x}_{<j})$. 
In the following, we use $\rho$ to denote density functions of random variables, where disambiguation will be obvious from context. 
Define $\bar{\bsx}_l = \proj_{A} \bsx_l$ for each $l\in[j]$.
We have for some $c\in\re$,
\begin{align}\label{eq:bayes}
    \rho(\bsA | \bsx_{\leq j}, \tilde{\bsx}_{<j}, \bsbarA) = \frac{\rho(\bsbarA, \bsx_{\leq j}, \tilde{\bsx}_{<j} | \bsA)\rho(\bsA)}{\rho(\bsbarA, \bsx_{\leq j}, \tilde{\bsx}_{<j})} = c\cdot \rho(\bsbarA, \bsx_{\leq j}, \tilde{\bsx}_{<j} | \bsA).
\end{align}
For the second inequality, note that
since the prior distribution of $A$ is uniform, $\rho(\bsA)$ does not depend on the value of $\bsA$. Further, $\rho(\bsx_{\leq j}, \tilde{\bsx}_{<j}, \bsbarA)$ also does not depend on $\bsA$. What remains is to show that $\rho(\bsbarA, \bsx_{\leq j}, \tilde{\bsx}_{<j} | \bsA) = c' \cdot \ind\{A\in\cA\}$ for some $c' \in\re$. 
We have,
\begin{align} \label{eq:likelihood}
    \rho(\bsbarA, \bsx_{\leq j}, \tilde{\bsx}_{<j} | \bsA) 
    = \frac{\rho(\bsbarA, \bsx_{\leq j}, \tilde{\bsx}_{<j} | \bsA)}{\rho(\bsx_{\leq j}, \tilde{\bsx}_{<j} | \bsA)} \rho(\bsx_{\leq j}, \tilde{\bsx}_{<j} | \bsA),
\end{align}
and,
\begin{align*}
    \rho(\bsx_{\leq j}, \tilde{\bsx}_{<j} | \bsA) &= \rho(\bsx_j |  \bsx_{< j},\tilde{\bsx}_{<j}, \bsA)\rho(\tilde{\bsx}_{j-1} | \bsx_{< j}, \tilde{\bsx}_{<j-1}, \bsA) \rho(\bsx_{< j}, \tilde{\bsx}_{<j-1} | \bsA) \\
    &~~\vdots \\
    &= \prod_{l=1}^j \rho(\bsx_l |  \bsx_{< l},\tilde{\bsx}_{<l}, \bsA)\rho(\tilde{\bsx}_{l-1} | \bsx_{< l}, \tilde{\bsx}_{<l-1}, \bsA) \\
    &= \prod_{l=1}^j \rho(\bsx_l | \tilde{\bsx}_{<l})\rho(\tilde{\bsx}_{l-1} | \bsx_{l-1}, \bsA) \\
    &= \prod_{l=1}^j \rho(\bsx_l | \tilde{\bsx}_{<l})\frac{\rho(\tilde{\bsx}_{l-1}, \bsx_{l-1} | \bsA)}{\rho(\bsx_{l-1}| \bsA)}. %
\end{align*}
Plugging into Eqn. \eqref{eq:likelihood} we obtain,
\begin{align*}
    \rho(\bsbarA, \bsx_{\leq j}, \tilde{\bsx}_{<j} | \bsA) = \frac{\rho(\bsbarA, \bsx_{\leq j}, \tilde{\bsx}_{<j} | \bsA)}{\rho(\bsx_{\leq j}, \tilde{\bsx}_{<j} | \bsA)} \prod_{l=1}^j \rho(\bsx_l | \tilde{\bsx}_{<l})\frac{\rho(\tilde{\bsx}_{l-1}, \bsx_{l-1} | \bsA)}{\rho(\bsx_{l-1}| \bsA)}.
\end{align*}
Since for any $l\in[j-1]$, because $\tilde{x}_l$ is completely determined by $x_{l}$ and $A$, the term $\frac{\rho(\tilde{\bsx}_{l}, \bsx_{l} | \bsA)}{\rho(\bsx_{l}| \bsA)}$ is $1$ if and only $\proj_A^\perp \bsx_{l} = \tilde{\bsx}_l$ and is $0$ otherwise; recall the orthogonal projection of a vector onto a linear subspace is unique.
Similarly, $\bar{A}$ (and $\bar{x}_{<j}$) are completely determined by $x_{\leq j}$ and $A$.
Assuming each $\tilde{\bsx}_{l} = \proj_{A}^\perp \bsx_l$ for all $l\in[j-1]$ (i.e. the case where $\frac{\rho(\tilde{\bsx}_{l}, \bsx_{l} | \bsA)}{\rho(\bsx_{l}| \bsA)}$ is non-zero), then the term $\frac{\rho(\bsbarA, \bsx_{\leq j}, \tilde{\bsx}_{<j} | \bsA)}{\rho(\bsx_{\leq j}, \tilde{\bsx}_{<j} | \bsA)}$ is $1$ if and only if $\forall i\in[d']: \bar{\bsa}_i = \proj_{\bar{\bsx}_{<j}} \bsa_i$ and is $0$ otherwise ($\bar{\bsa}_i$ is the $i$'th row of $\bsbarA$). 
Since each $\rho(\bsx_l|\tilde{\bsx}_{<l})$ term does not depend on the choice of $\bsA$, via Eqn. \eqref{eq:bayes} this shows that the conditional distribution of $A$ is uniform over $\cA'$, where $\cA'$ is the set of matrices, $\bsA$, such that $\bsA\in\Supp(A)$ and for every $l\in[j-1]$, $ \proj_{\bsA}^\perp x_l = \tilde{x}_l$ and $\proj_{\bar{x}_{<j}} \bsa_i = \bar{a}_i$ for every $i\in[d']$. 

What remains is to show that $\cA'=\cA$.
Let $\bsa_{1:d'}$ denote some potential value for the rows of $A$. 
Recall we are conditioning on a realization of $\tilde{x}_{< j}$, $x_{\leq j}$, and $\{\bar{a}_i\}_{i\in[d']}$ in their joint support, which also determines $\bar{x}_{<j}$. 
Both $\cA$ and $\cA'$ require $\proj_{\bar{x}_{<j}} \bsa_i = \bar{a}_i$ for all $i\in[d']$. 
The condition $\|\proj_{\bar{x}_{<j}}^\perp \bsa_i\| = \sqrt{1-\|\bar{a}_i\|^2}$ follows from that fact that that $\Supp(A) = (\cS(1))^{d'}$.
Let $\tilde{\bsa}_i = \proj_{\bar{x}_{<j}^\perp}\bsa_i$ for all $i\in[d']$.
What remains is to show 
$\proj_{\bsA}^\perp x_l = \tilde{x}_l$ for all $l\in[j-1]$ if and only if $\tilde{\bsa}_1,...,\tilde{\bsa}_{d'} \in  \Span^\perp(\tilde{x}_{<j})$.

We have for any $l \in [j-1]$ and $i\in[d']$,
\begin{align*}
    \ip{\tilde{x}_l}{\bsa_i} 
    &= \ip{x_l}{\bsa_i} - \ip{\bar{x}_l}{\bsa_i} \\
    &=  \ip{x_l}{\proj_{\bar{x}_{<j}} \bsa_i} + \ip{x_l}{\proj_{\bar{x}_{<j}}^\perp \bsa_i} - \ip{\bar{x}_l}{\bsa_i} \\
    &\overset{(i)}{=} \ip{x_l}{\proj_{\bar{x}_{<j}}^\perp \bsa_i} \\
    &= \ip{\bar{x}_l}{\proj_{\bar{x}_{<j}}^\perp \bsa_i} + \ip{\tilde{x}_l}{\proj_{\bar{x}_{<j}}^\perp \bsa_i} \\
    &= \ip{\tilde{x}_l}{\proj_{\bar{x}_{<j}}^\perp \bsa_i} \\
    &= \ip{\tilde{x}_l}{\tilde{\bsa}_i}.
\end{align*}
Step $(i)$ uses the fact that for any $l<j$, $\ipnos{x_l}{\proj_{\bar{x}_{<j}} \bsa_i} = \ipnos{\proj_{\bar{x}_{<j}}(\bar{x}_l + \tilde{x}_l)}{ \bsa_i} = \ip{\bar{x}_l}{\bsa_i}$, noting $\tilde{x}_l$ is orthogonal to $\Span(\bar{x}_{<j})$ for any realization in their (joint) support.
The above shows the vector $\tilde{x}_l$ is orthogonal to $\Span(\bsA)$ if and only if $\Span(\tilde{\bsa}_{1:d'})$ is orthogonal to $\tilde{x}_l$.
Since the orthogonal projection of $x_l$ onto $\Span^\perp(\bsA)$ is the closest vector to $x_l$ which is orthogonal to $\bsA$, if $\Span(\tilde{\bsa}_{1:d'}) \in \Span^\perp(\tilde{x}_l)$ then
$\proj_A^\perp x_l = \tilde{x}_l$ since $\tilde{x}_l$ is in the set of vectors orthogonal to $\bsA$.
Alternatively, if $\Span(\tilde{\bsa}_{1:d'}) \notin \Span^\perp(\tilde{x}_l)$, then $\proj_{\bsA}^\perp x_l \neq \tilde{x}_l$ since $\tilde{x}_l$ is not orthogonal to $\Span(\bsA)$.
We have thus shown that $\proj_{\bsA}^\perp x_l=\tilde{x}_l$ for all $l\in[j-1]$ if and only if $\tilde{\bsa}_{1:d'} \in \Span^\perp(\tilde{x}_{<j})$, and we are done.
\end{proof}

Since Lemma \ref{lem:opt-ocvg} uses an OCVG matrix, $A'$, which is only ``part of'' the matrix $A$ used in the barrier function, we need to ensure that a query which is not orthogonal to $A'$ is also not too orthogonal to $A$ This is accomplished by the following lemma.
\begin{lemma} \label{lem:Aprime-lb-A}
Let $A\in\re^{d'\times d}$ have each row $a_i \sim \Unif(\cS(1))$, $\forall i\in[d']$.
For any linear subspace of $\Span(A)$, $\cS$, let $A'$ be the matrix formed by projecting each row of $A$ onto $\cS$.
With probability at least $1-2^{-\Omega(d)}$ under the draw of $A$, both of the following hold for any $v\in\re^d$,
\begin{enumerate}
\item $\|Av\|_\infty  \geq \frac{\|A'v\|_\infty}{2 d}$

\item For any $x\in\cB(1)$,
$|\ip{v}{x} - \ip{v}{\proj_A^\perp x}| \leq 2d\|Av\|_\infty$

\end{enumerate}
\end{lemma}
\begin{proof}
Let $v\in\re^d$. Let $v_A = \proj_{A} v$ and let $\gamma \in \re^{d'}$ be such that $v = \gamma^T A$. By Holder's inequality we have
$\|v_A\|^2 = \|\sum_{j=1}^{d'} \gamma_j \ip{a_j}{v_A}\| \leq \|\gamma\|_1 \|A v_A\|_\infty = \|\gamma\|_1 \|A v\|_\infty$. 
Thus by Lemma \ref{lem:l1-coeff-norm} we have,
\begin{align}\label{eq:A-comp-via-max}
    \frac{\|v_A\|^2}{\|Av\|_\infty} \leq \|\gamma\|_1 \leq \frac{\sqrt{d}\|v_A\|}{\sigma_{min}(A)} \implies \frac{\|v_A\|\sigma_{min}}{\sqrt{d}} \leq \|Av\|_\infty,
\end{align}
where $\sigma_{min}$ denotes the smallest non-zero singular value of $A$. 
Since $\|v_A\| \geq \|\proj_{S} v\| \geq \|A' v\|_\infty$, we have,
\begin{align*}
    \|Av\|_\infty \geq \frac{\|v_A\|\sigma_{min}}{\sqrt{d}} \geq \frac{\|A'v\|_\infty \sigma_{min}}{\sqrt{d}}. %
\end{align*}
It remains to show that w.h.p. $\sigma_{min} \geq \frac{1}{2\sqrt{d}}$.
Recall $A$ is sampled as a $d'\times d$ matrix whose rows are uniformly random unit vectors.
By standard result for random projection we have for any $j\in [d']$,
\begin{align*}
    \pr{A}{\|\proj_{a_{\neq j}}^\perp a_j\| \geq \frac{\|a_j\|}{2} } \leq 1-e^{-\Omega(d)}.
\end{align*}
Thus by a union bound a similar event holds for every $j\in [d']$ w.p. at least $1-e^{-\Omega(d)}$. The lower bound on $\sigma_{min}$ then follows from Lemma \ref{lem:singular-vals}. 

For the second part of the lemma, when $\sigma_{min} \geq \frac{1}{2\sqrt{d}}$ we also have by the previously established Eqn. \eqref{eq:A-comp-via-max} that,
\begin{align*}
    |\ip{v}{x} - \ipnos{v}{\proj_A^\perp x}|  = |\ip{v}{\proj_A x}| \leq \|\proj_A v\| \leq 2d\|Av\|_\infty.
\end{align*}
\end{proof}

\subsection{Proof of Lemma \ref{lem:resisting-consistent}} \label{app:resisting-consistent}
\paragraph{Oracle implementation.} 
Before proving the lemma, we define the exact implementations of the resisting and true oracles. For convenience, we recall the definition of the objective function.
\begin{align*}
    F(w) = \max\big\{\max_{j\in[N]}\bc{|\ip{w}{x_j} - 2\alpha| - j\gamma},~ d^{6}\|A w\|_\infty,~ \max_{z\in\cZ}\bc{\ip{w}{z}} - 4\alpha \big\}.
\end{align*}
For $w\in\re^d$, we have the true oracle, $\cO$, select an element from $\partial F(w)$ in the following way. If $\nem(w)$ is largest, $\cO$ returns the lowest index Nemirovski vector in the subgradient of $\nem(w)$ and similarly for $\tilde{\cO}_j$.
When $\barr(w)$ is largest, $\cO$ returns the lowest index row of $A$ in the subgradient, and similarly for the resisting oracle.
When $\wall(w)$ is largest, we use $\argmax_{z\in\cZ}\{\ip{w}{z}\}$. Ties can be broken by choosing the element in the maximal set closest to some arbitrary fixed vector $z^*\in\cS(1)$. The analogous strategy is used for the resisting oracle. When more than one of $\nem,\barr$ or $\wall$ is maximal, we take the oracle to prefer a subgradient of the $\nem$ first, then $\barr$, then $\wall$. Note this means that when $F(w)=\max\{\nem_j(w),\barr(w)\}$ then the true and resisting oracles return the same response.

\begin{lemma*}
For Algorithm \ref{alg:opt}, 
$\pr{A,x_{1:N}}{\forall j\in[N-1], t\in[T]: \tilde{\cO}_j(\ujt) = \cO(\ujt)} \geq 1-O(\frac{1}{d})$.
\end{lemma*}
\begin{proof}
For notation, define $\nem_j(w) = \max_{i\in[j]}\bc{|\ip{w}{x_i} - 2\alpha| - i\gamma}$ and $\wall_j(w) = \max_{z\in\cZ_j}\bc{\ip{w}{z}} - 4\alpha$.

To show that the resisting-oracle is consistent with high probability,
we will condition on, 
\begin{align} \label{eq:cases-ocvg}
  \forall j\in[N-1], t\in[T], ~~~ \|A\tujt\|_\infty \geq \frac{\|\tujt\|}{8d^{5.5}} ~\text{ or }~  \max\{|\ip{\tujt}{x_j}|,|\ip{\tujt}{\tilde{x}_j}|\} \leq 6\|\tujt\|\sqrt{\frac{k}{d}},  
\end{align}
and,
\begin{align} \label{eq:cases-subgaussian}
    \forall j\in[N-1], t\in[T],j' \geq j+1  \quad~ |\ip{\ujt}{x_{j'}}| \leq 3\|\tujt\|\sqrt{\frac{\log(d)}{d}} \leq \|\tujt\| \sqrt{\frac{k}{d}},
\end{align}
and that $\|\tilde{x}_i\| \geq 1/2$~ for every $i\in[N]$.
Note $\ip{\ujt}{x_l}=\ip{\tujt}{x_l}$. 
By Lemma \ref{lem:opt-ocvg} and a standard concentration results (Lemmas \ref{lem:random-correlation} and \ref{lem:random-projection}) this happens with probability at least $1-O\big(\frac{N}{d^3}\big) - O\big(\frac{1}{d}\big) - 2^{\Omega(d)} = 1-O(\frac{1}{d})$.

Conditional on this event, we will consider several cases and show that the true and resisting oracles match in each. %
We note the following uses two lemmas, deferred to Section \ref{sec:lemmas-for-opt-reduction} below, which handle the consistency of the wall function. Throughout the following, note that under our choice of the oracle implementations, whenever $F(\ujt) = \max\{\nem_j(\ujt),\barr(\ujt)\}$, the resisting and true oracles return the same response.

\paragraph{Case 1: $\|\tujt\| \in [0,\alpha\sqrt{\frac{k}{d}}] ~~\&~~ \|A\tujt\|_\infty \geq \frac{\|\tujt\|}{8d^6}$.}
Since we have conditioned on $\|\tilde{x}_i\| \geq 1/2$ for every $i\in[N]$, by Lemma \ref{lem:barrier-beats-wall}, we have that $F(\ujt)=\max\{\nem(\ujt),\barr(\ujt)\}$ since $\|A\tujt\|_\infty \geq \frac{\|\tujt\|}{8d^6}$. Thus it only remains to show $\nem(\ujt) = \nem_j(\ujt)$.
Observe that for any $j' \geq j+1$,
\begin{align*} 
    |\ip{\ujt}{x_{j}}-2\alpha| - j\gamma \leq |\ip{\ujt}{x_{j'}}-2\alpha| - j'\gamma \implies 10\alpha\sqrt{\frac{k}{d}} \leq |\ip{\ujt}{x_{j}}| + |\ip{\ujt}{x_{j'}}|. 
\end{align*}
Since we have conditioned on the event, $\forall j'\geq j+1: |\ip{\ujt}{x_{j'}}| \leq 3\|\tujt\|\sqrt{\frac{\log(d)}{d}} \leq \|\tujt\| \sqrt{\frac{k}{d}}$, we thus have,
\begin{align} \label{eq:nem-needs-large-ip}
     |\ip{\ujt}{x_{j}}| \leq 6\alpha\sqrt{\frac{k}{d}} \implies \nem(\ujt) = \nem_j(\ujt)
\end{align}
Since $\|\tujt\|\leq\alpha\sqrt{\frac{k}{d}}$, Cauchy-Schwarz implies the condition of the above is satisfied, and thus $F(\ujt)=\max\{\nem_j(\ujt),\cR(w)\}$, which implies the resisting and true oracles provide the same response.

\paragraph{Case 2: $\|\tujt\| \in [\alpha\sqrt{\frac{k}{d}},1] ~~\&~~ \|A\tujt\|_\infty \geq \frac{\|\tujt\|}{8d^6}$.} 
As in the previous case,
by Lemma \ref{lem:barrier-beats-wall}, we have that $F(\ujt)=\max\{\nem(\ujt),\barr(\ujt)\}$ since $\|A\tujt\|_\infty \geq \frac{\|\tujt\|}{8d^6}$. 

Since $\|\tujt\|\geq\alpha\sqrt{\frac{k}{d}}$,  we have 
$\barr(\ujt) \geq \frac{\sqrt{d}}{8}\|\tujt\|\geq \alpha \sqrt{k}/8 \geq \alpha \sqrt{32\log(d)}/8$. 
Thus for $d$ larger than some constant we have
$\barr(\ujt) \geq 3.5\max\{\|\tujt\|,\alpha\}$.
However, 
for any $i\geq j$, we have $|\ip{\ujt}{x_i}-2\alpha|-i\gamma < 3\max\{\|\tujt\|,\alpha\}$.
Thus we have $F(\ujt) = \max\{\nem_j(\ujt),\barr(\ujt)\}$ and the oracles match.

\paragraph{Case 3: $\|\tujt\| \in [0,3\alpha] ~~\&~~  \|A\tujt\|_\infty \leq \frac{\|\tujt\|}{8d^6}$. } 
Because $\|A\tujt\|_\infty \leq \frac{\|\tujt\|}{8d^6}$, Eqn. \eqref{eq:cases-ocvg} implies $\max\{|\ip{\tujt}{x_j}|, |\ip{\tujt}{\tilde{x}_j}|\} \leq 6\|\tujt\|\sqrt{\frac{k}{d}}$.
Lemma \ref{lem:wall} thus implies the true and resisting oracle return the same vector if $\wall$ is maximal. 
Recall we have also conditioned on the event that for all $l\geq j+1$, $|\ip{\tujt}{\tilde{x}_j}| \leq 6\|\tujt\|\sqrt{\frac{k}{d}}$.
Thus by Eqn. \eqref{eq:nem-needs-large-ip}
we have that $\nem(\ujt) = \nem_j(\ujt)$, and the oracles return the same response.

\paragraph{Case 4: $\|\tujt\| \in [3\alpha,1] ~~\&~~ \|A\tujt\|_\infty \leq \frac{\|\tujt\|}{8d^6}$.} 
Since $\|A \tujt\|_\infty \leq \frac{\|\tujt\|}{8d^6}$, by Eqn. \eqref{eq:cases-ocvg} it must be that $\max\{|\ip{\tujt}{x_j}|, |\ip{\tujt}{\tilde{x}_j}|\} \leq 6\|\tujt\|\sqrt{\frac{k}{d}}$.
Now by Lemma \ref{lem:wall}, we have $\wall(\ujt) \geq \|\tujt\| - 4\alpha \geq 0$,
which implies $\max\{\wall(\ujt),\nem_j(\ujt)\} \geq \nem(\ujt)$.
Further, Lemma \ref{lem:wall} the resisting and true oracle return the same response if the wall function is maximal. This implies the resisting oracle is consistent with the true oracle.
\end{proof}

\subsubsection{Sublemmas for the Proof of Lemma \ref{lem:resisting-consistent}} \label{app:resist-consist-lemmas}
The first two lemmas in this subsection handle the wall function depending on the outcome of the OCVG. 
Recall,
\begin{align*}
    F_j(w)\! =\! \max\big\{\max_{i\in[j]}\bc{|\ip{w}{x_i} - 2\alpha| - i\gamma},~ d^{6}\|A w\|_\infty,~ \max_{z\in\cZ_j}\bc{\ip{w}{z}}\! -\! 4\alpha \big\},
\end{align*}
where $\cZ_j = \big\{z: \|z\|=1 ~~\land~~ z \notin \bigcup_{i\in[j-1]} \cC_{i}\big\}$ and for $i\in[N]$, $\mathcal{C}_i = \big\{y: \abs{\ip{\tilde{x}_i}{y}} > 6\sqrt{\frac{k}{d}}\big\}$.
We start with the case where the optimizer does not find a vector substantially correlated with the problem vector.
This analysis is similar to that in \cite{bubeck19_highlyparallel}.
\begin{lemma} \label{lem:wall}
Let $j\in[N]$.
Let $w\in\cB(1)$ be such that ~$\forall i \geq j$, $|\ipnos{\proj_{\tilde{x}_{< j}}^\perp w}{\tilde{x}_i}| \leq 6\|\proj_{\tilde{x}_{< j}}^\perp w\| \sqrt{\frac{k}{d}}$. 
Then, if $\wall(w) > \max\{\nem(w),\barr(w)\}$, it holds that $\cO(w)=\tilde{\cO}_j(w)$ and further that
$\wall(w) \geq \|\proj_{\tilde{x}_{<j}}^\perp w\| - 4\alpha.$
\end{lemma}
\begin{proof}
Via an orthogonal decomposition, we have for any $w$ that,
\begin{align} \label{eq:wallone}
    \wall(w)=\max_{\substack{a^2+b^2=
    1 \\ g+y\in\cZ, g\in\cG_a, y\in\cY_b}}\bc{\ip{g+y}{w}} - 4\alpha,
\end{align}
where $\cG_a=\Span(\tilde{x}_{<j})\cap \cS(a)$ and $\cY_b=\Span^\perp(\tilde{x}_{<j})\cap\cS(b)$. 
Alternatively, the resisting oracle uses the function (noting $g+y \in \cZ_j \iff g\in\cZ_j$),
\begin{align}\label{eq:walltwo}
    w \mapsto \max_{\substack{a^2+b^2=
    1 \\ g\in\cZ_j, g\in\cG_a, y\in\cY_b}}\bc{\ip{g+y}{w}} - 4\alpha.
\end{align}
Further, for any $y\in\cY_b$ it holds that $g+y\in\cZ \implies g\in\cZ_j$. %
Thus it suffices to show that for any $g\in\cZ_j \cap \cG_a$, the maximizing value of $y$ in both Eqns. \eqref{eq:wallone} and \eqref{eq:walltwo} is  $y=b\frac{\proj_{\tilde{x}_{<j}}^\perp w}{\|\proj_{\tilde{x}_{<j}}^\perp w\|}$, and thus the oracle responses are the same (assuming as we have that the resisting and true oracles choose the same $g\in \cZ_j\cap \cG_a$).
For any setting of $g\in\cZ_j \cap \cG_a$, the setting of $y'\in \cY_b$ (i.e. ignoring the condition that $g+y\in\cZ$) which maximizes $\ip{w}{g+y}$ is $b\frac{\proj_{\tilde{x}_{<j}}^\perp w}{\|\proj_{\tilde{x}_{<j}}^\perp w\|}$. 
But by the assumption that $\forall i \geq j$, $|\ipnos{\proj_{\tilde{x}_{< j}}^\perp w}{\tilde{x}_i}| \leq 6\|\proj_{x_{< j}}^\perp w\| \sqrt{\frac{k}{d}}$, such a setting of $y$ still yields $g+y\in\cZ$ whenever $g\in\cZ_j \cap \cG_a$, as desired. 

Finally, when $y=b\frac{\proj_{\tilde{x}_{<j}}^\perp w}{\|\proj_{\tilde{x}_{<j}}^\perp w\|}$, we have $\cW(w) = \ip{g}{w} + b\|\proj_{\tilde{x}_{<j}}^\perp w\| -4\alpha$ for some $g$ of norm $a=\sqrt{1-b^2}$.
The lower bound comes from the fact that $\cW$ maximizes over all $a^2+b^2=1$ and $y\in\cY_b$, thus the value of $\cW$ must be lower bounded by the case where $a=0$, $b=1$, and $y=b\frac{\proj_{\tilde{x}_{<j}}^\perp w}{\|\proj_{\tilde{x}_{<j}}^\perp w\|}$.
\end{proof}

The next lemma handles the wall function in the case where the barrier function is large, which happens if the optimizer finds a point sufficiently correlated with the problem vector.
\label{sec:lemmas-for-opt-reduction}
\begin{lemma}\label{lem:barrier-beats-wall}
Assume $\|\tilde{x}_i\| \geq 1/2$~ for all $i\in[N]$. Then if 
$\|A\tujt\|_\infty \geq \frac{\|\tujt\|}{8d^{5.5}}$, then 
$\max\{\nem(\ujt),\barr(\ujt)\} > \wall(\ujt)$.
\end{lemma}
\begin{proof}
First we show an upper bound on $\wall(\ujt)$. Let $\nu = \max\limits_{i<j}\{\ip{\ujt}{\tilde{x}_i}\}$. Since $\tilde{x}_{1:j}$ is an orthogonal set of vectors each of norm at least $1/2$, $\|\proj_{\tilde{x}_{<j}} \ujt\| \leq 2\sqrt{j}\nu$ and we have,
\begin{align}\label{eq:wall-bound-via-norm}
    \wall(\ujt) \leq \sqrt{\|\proj_{\tilde{x}_{<j}}\ujt\|^2 + \|\tujt\|^2} - 4\alpha \leq 2\nu\sqrt{j/d} + \|\tujt\| - 4\alpha \leq \frac{\nu}{2} + \|\tujt\| - 4\alpha.
\end{align}
We consider two cases based on whether $\nu \leq 4\|\tujt\|$. 
If $\nu \leq 4\|\tujt\|$, then Eqn. \eqref{eq:wall-bound-via-norm} implies 
$\wall(\ujt) \leq 3\|\tujt\|-4\alpha$.
On the other hand, because we have assumed $\|A\tujt\|_\infty \geq \frac{\|\tujt\|}{8d^{5.5}}$, we have 
$\barr(\ujt) \geq \frac{\sqrt{d}}{8}\|\tujt\| > 3\|\tujt\|$ (for $\sqrt{d}>24$). Thus we have $\barr(\ujt) \geq \wall(\ujt)$. 

In the other case where $\nu\geq 4\|\tujt\|$, we consider $\nem(\ujt)$. 
Then there exists $i\in[j-1]$ such that,
\begin{align*}
    \ip{\ujt}{x_i} &= \ip{\ujt}{\proj_A x_i} + \ip{\ujt}{\tilde{x}_i} \\
    &= \ipnos{\proj_{\tilde{x}_{<j}}\ujt}{\proj_A x_i} + \ipnos{\proj_{\tilde{x}_{<j}}^\perp\ujt}{\proj_A x_i} + \ip{\ujt}{\tilde{x}_i} \\
    &= \ipnos{\tujt}{\proj_A x_i} + \ip{\ujt}{\tilde{x}_i} 
    \geq \nu - \|\tujt\| \geq \frac{3\nu}{4}.
\end{align*}
This implies $\nem(\ujt) \geq \frac{3\nu}{4}-3\alpha$.
But by Eqn. \eqref{eq:wall-bound-via-norm} we have $\wall(\ujt) \leq \frac{3\nu}{4}-4\alpha$.
Thus $\nem(\ujt) > \wall(\ujt)$.
\end{proof}

\paragraph{Acknowledgements.}
M. Menart and A. Nikolov are supported by an NSERC Discovery Grant (RGPIN-2021-03206), and the Canada Research Chairs program (CRC-2020-00004).

\bibliographystyle{alpha} 
\bibliography{ref-base,ref-add}
\end{document}